\documentclass{article}

\usepackage[nonatbib,final]{neurips_2021}
\usepackage[numbers]{natbib}

\usepackage[utf8]{inputenc} 
\usepackage[T1]{fontenc}    
\usepackage{hyperref}       
\usepackage{url}            
\usepackage{booktabs}       
\usepackage{amsfonts}       
\usepackage{nicefrac}       
\usepackage{microtype}      
\usepackage{xcolor}         

\usepackage{soul}

\usepackage{xcolor}
\usepackage[above]{placeins}

\usepackage{subcaption}
\usepackage{enumitem}

\usepackage{xfunctions}
\input{aliases}

\usepackage{thmtools,thm-restate}

\declaretheorem{proposition}
\declaretheorem{lemma}
\newtheorem*{lemma*}{Lemma}

\newtheorem{definition}{Definition}

\usepackage[ruled,linesnumbered]{algorithm2e}
\usepackage{etoolbox}
\makeatletter
\patchcmd{\@algocf@start}
  {-1.5em}
  {0pt}
  {}{}
\makeatother

\SetCommentSty{mycommfont}

\title{Two-sided fairness in rankings via Lorenz dominance}

\author{
    Virginie Do\textsuperscript{\rm 1,2}, Sam Corbett-Davies\textsuperscript{\rm 1}, Jamal Atif\textsuperscript{\rm 2}, Nicolas Usunier\textsuperscript{\rm 1}\\
    \textsuperscript{\rm 1}Facebook AI\\
    \textsuperscript{\rm 2}LAMSADE, Université PSL, Université Paris Dauphine, CNRS, France\\
    \texttt{virginiedo@fb.com, scd@fb.com, jamal.atif@dauphine.psl.eu, usunier@fb.com}
}

\begin{document}

\maketitle

\begin{abstract}
We consider the problem of generating rankings that are fair towards both users and item producers in recommender systems. We address both usual recommendation (e.g., of music or movies) and reciprocal recommendation (e.g., dating). Following concepts of distributive justice in welfare economics, our notion of fairness aims at increasing the utility of the worse-off individuals, which we formalize using the criterion of \emph{Lorenz efficiency}. It guarantees that rankings are Pareto efficient, and that they maximally redistribute utility from better-off to worse-off, at a given level of overall utility. We propose to generate rankings by maximizing concave welfare functions, and develop an efficient inference procedure based on the Frank-Wolfe algorithm. We prove that unlike existing approaches based on fairness constraints, our approach always produces fair rankings. Our experiments also show that it increases the utility of the worse-off at lower costs in terms of overall utility.\end{abstract}

\section{Introduction}

Recommender systems have a growing impact on the information we see and on our life opportunities, as they help us browse news articles, find a new job, house, or people to connect with. While the objective of recommender systems is usually defined as maximizing the quality of recommendations from the user's perspective, the recommendations also have an impact on the recommended ``items''. News outlets rely on exposure to generate revenue, finding a job depends on which recruiter gets to see our resume, and the effectiveness of a dating application also depends on who we are recommended to---and if we are being recommended, then someone else is not. 
\emph{Two-sided fairness in rankings} is the problem of generating personalized recommendations by fairly mediating between the interests of 
users and items. It involves a complex multidimensional trade-off. Fairness towards item producers requires boosting the exposure of small producers (e.g., to avoid winner-take-all effects and popularity biases \citep{abdollahpouri2019unfairness}) at the expense of average user utility. Fairness towards users aims at increasing the utility of the least served users (e.g., so that least served users do not support the cost of item-side fairness), once again at the expense of average user utility. The goal of this paper is to provide an algorithmic framework to generate rankings that achieve a variety of these trade-offs, leaving the choice of a specific trade-off to the practitioner.
%


The leading approach to fairness in rankings is to maximize user utility under constraints of equal item exposure (or equal quality-weighted exposure) \citep{singh2018fairness,biega2018equity} or equal user satisfaction \citep{basu2020framework}. When these constraints imply an unacceptable decrease in average user utility, so-called ``trade-offs between utility and fairness'' \citep{zehlike2020reducing,mehrotra2018towards} are obtained by relaxing the fairness constraints, leading to the optimization of a trade-off between average user utility and a measure of users' or items' inequality. 

Thinking about fairness in terms of optimal utility/inequality trade-offs has, however, two fundamental limitations. First, the optimization of a utility/inequality trade-off is not necessarily Pareto-efficient from the point of view of users and items: it sometimes chooses solutions that decrease the utility of some individuals without making anybody else better off. We argue that reducing inequalities by decreasing the utility of the better-off is not desirable if it does not benefit anyone. The second limitation 
is that focusing on a single measure of inequality does not address the question of how inequality is reduced, and in particular, which fraction of the population benefits or bears the cost of reducing inequalities.


In this paper, we propose a new framework for two-sided fairness in rankings grounded in the analysis of generalized Lorenz curves of user and item utilities. Widely used to study efficiency and equity in cardinal welfare economics \citep{shorrocks1983ranking}, these curves plot the cumulative utility obtained by fractions of the population ordered from the worst-off to the best-off. A curve that is always above another means that all fractions of the populations are better off. We define fair rankings as those with non-dominated generalized Lorenz curves for users and items. First, this definition guarantees that fair rankings are Pareto efficient. Second, examining the entirety of the generalized Lorenz curves provides a better understanding of which fractions of the population benefit from an intervention, and which ones have to pay for it. We present our general framework based on Lorenz dominance in usual recommendation settings (e.g., music or movie recommendation), and also show how extend it to \emph{reciprocal recommendation} tasks such as dating applications or friends recommendation, where users are recommended to other users. 

We present a new method for generating rankings based on the maximization of concave welfare functions of users' and items' utilities. The parameters of the welfare function control the relative weight of users and items, and how much focus is given to the worse-off fractions of users and items. We show that rankings generated by maximizing our welfare functions are fair for every value of the parameters.  Our framework does not aim at defining what parameters are suitable in general — rather, the choice of a specific trade-off depends on the application.

From an algorithmic perspective, two-sided fairness is challenging because items’ utilities depend on the rankings of all users, requiring global inference. Previous work on item-side fairness addressed this issue with heuristic methods without guarantees or control on the achievable trade-offs. We show how the Frank-Wolfe algorithm can be leveraged to make inference tractable, addressing both our welfare maximization approach and existing item-side fairness penalties.

We demonstrate that our welfare function approach enjoys stronger theoretical guarantees than existing methods. While it always generates rankings with non-dominated generalized Lorenz curves, many other approaches do not. We show that one of the main criteria of the literature, called equity of attention by \citet{biega2018equity}, can lead to decrease user utility, while \emph{increasing} inequalities of exposure between items. Moreover, equal user satisfaction criteria in reciprocal recommendation can lead to decrease the utility of \emph{every user}, even the worse-off. Our notion of fairness prevents these undesirable behaviors. We report experimental results on music and friend recommendation tasks, where we analyze the trade-offs obtained by different methods by looking at different points of their Lorenz curves. Our welfare approach generates a wide variety of trade-offs, and is, in particular, more effective at improving the utility of worse-off users than the baselines.

We present our formal framework in Section~\ref{sec:onesided}. We discuss the theoretical properties of previous approaches in Section~\ref{sec:fairness_exposure}, and present our ranking algorithm in Section~\ref{sec:algo}. Our experiments are described in Section~\ref{sec:experiments}, and the related work is discussed in Section \ref{sec:related}.

\section{Two-sided fairness via Lorenz dominance}\label{sec:onesided}

\subsection{Formal framework}\label{sec:framework}

\paragraph{Terminology and notation.} 
We identify an item with its producer, so that ``item utility'' means ``item producer's utility''.  The main paper focuses on fairness towards individual users and items. We describe in Appendix~\ref{sec:groups} the extension of our approach to sensitive groups of users or items. $\card{\cal X}$ denotes the cardinal of the set $\cal X$. Given $n\in\mathbb{N}$, we denote by $\intint{n}=\xSet{1, \ldots, n}$. 
The set of users $\userS$ is identified with $\{1,..., \card{\userS}\}$ and the set of items $\itemS$ is identified with $\{\card{\userS}+1,...,\n\}$ where  $\n = \card{\userS} + \card{\itemS}$.
For $(i,j)\in \userS\times \itemS$, we denote by $\muij$ the value of item $j$ to user $i$.

A (deterministic) ranking 
$\sigma:\itemS\to\intint{\nitems}$ is a one-to-one mapping from items $j$ to their rank $\sigma(j)$. Following \citep{singh2018fairness}, we use \emph{stochastic rankings} because they allow us to perform inference using convex optimization (see Section~\ref{sec:algo}). 
The recommender system produces one stochastic ranking per user, represented by a 3-way \emph{ranking tensor} $\rrk$ where $\rrkijk$ is the probability that $j$ is recommended to $i$ at rank $k$. We denote by  $\rrkS$ the set of ranking tensors.

Utilities of users and items are defined through a position-based model, as in previous work~ \citep{singh2018fairness,biega2018equity,wu2021tfrom}. Let $\v\in\Re^{\nitems}$, where $\v_k$ is the exposure weight at rank $k$.
We assume that lower ranks receive more exposure, so that $\forall k\in\intint{\nitems-1}, \v_k\geq \v_{k+1}\geq 0$.\footnote{We use a user-independent $\v$ for simplicity. Considering user-dependent weights is straightforward.}
Given a user $i$ and a ranking $\sigma_i$, the \emph{user-side utility} of $i$ is the sum of the $\muij$s weighted by the exposure weight of their rank $\sigma_i(j)$: $u_i(\sigma_i)= \sum_{j\in\itemS} \v_{\sigma_i(j)} \muij$. Given an item $j$, the \emph{item-side utility} of $j$ is the sum over users $i$ of the exposure of $j$ to $i.$ These definitions extend to stochastic rankings by taking the expectation over rankings, written in matrix form:\footnote{We consider $\rrkij$ as a row vector in the formula, so that $\rrkij\v=\sum_{k=1}^\nitems \rrkijk\v_k$.} 
\begin{align}\label{eq:twosided}
\text{\emph{user-side utility:}} &\quad \tui(P)= \sum_{j\in\itemS} \muij \rrkij\v \quad&
\text{\emph{ item-side utility (exposure):}} &\quad \tuj(\rrk)=\sum_{i\in\userS} \rrkij\v
\end{align}
We denote by $\tuP = (\tuiP)_{i=1}^n$ the utility profile for $\rrk$, and by $\US=\xSet{\tuP:\rrk\in\rrkS}$ the set of feasible profiles. For $\bu\in\US$, $\buusers=(\bu_i)_{i\in\userS}$ and  $\tuitems=(\bu_i)_{i\in\itemS}$ denote the utility profiles of users and items respectively. 

\paragraph{Two-sided fairness in rankings.} In practice, values of $\muij$ are not known to the recommender system. Ranking algorithms use an estimate $\hmu$ of $\mu$ based on historical data. We address here the problem of \emph{inference}: the task is to compute the ranking tensor given $\hmu$, with the goal of making fair trade-offs between (true) user and item utilities. Notice that the user-side utility depends only on the ranking of the user, but for every item, the exposure depends on the rankings of \emph{all} users. Thus, accounting for both users' and items' utilities in the recommendations is a global inference problem.

\paragraph{More general item utilities} We consider exposure as the item-side utility to follow prior work and for simplicity. Our framework and algorithm readily applies in a more general case of \emph{two-sided preferences}, where items also have preferences over users (for instance, in hiring, job seekers have preferences over which recruiters they are recommended to). Denoting $\muji$ the value of user $i$ to item $j$, the item side-utility is then $\tuj(P)=\smash{\sum\limits_{i\in\userS}} \muji\rrkij\v$.

\begin{figure}[t]
    \centering
    \includegraphics[width=0.62\linewidth]{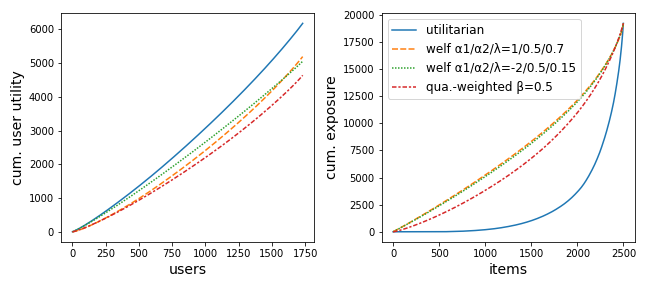}
    \quad\quad\quad\quad \includegraphics[width=0.27\linewidth]{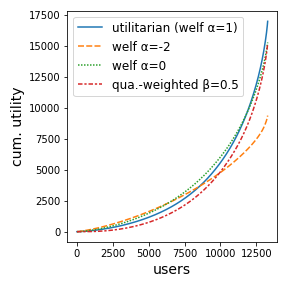}
    \caption{Generalized Lorenz curves for usual (left) and reciprocal (right) recommendation.}
    \label{fig:tradeoffs}
\end{figure}
\subsection{Lorenz efficiency and the welfare function approach}\label{sec:welfare}

Our notion of fairness aims at improving the utility of the worse-off users and items. Since this does not prescribe exactly which fraction of the worse-off users/items should be prioritized, the assessment of trade-offs requires looking at all fractions of the population. This is captured by the generalized Lorenz curve used in cardinal welfare economics \citep{shorrocks1983ranking}. 
Formally, given a utility profile $\bu$, let $(\bu_{(i)})_{i=1}^n$ be the sorted values in $\bu$ from smallest to largest, i.e., $\bu_{(1)} \leq \ldots \leq \bu_{(n)}$, then the generalized Lorenz curve plots $(\bucum_i)_{i=1}^n$ where $\bucum_i = \bu_{(1)}+\ldots+\bu_{(i)}$. To assess the fairness of trade-offs, we rely on the following dominance relations on utility profiles:

\textbf{Pareto-dominance $\pareto$.} $\bu\pareto\bv\iff (\forall i\in\intint{n}, \bu_i\geq\bv_i$ and $\exists i\in\intint{n}, \bu_i>\bv_i)$.

\textbf{Lorenz-dominance $\lorenz$.} Then $\bu\lorenz\bv \iff \bucum\pareto\bvcum$. 

We write $\lorenzw$ for non-strict Lorenz dominance (i.e., $\forall i, \bucum_i\geq \bvcum_i$). Notice that Pareto-dominance implies Lorenz-dominance. Our notion of fairness, which we call \emph{Lorenz efficiency}, states that a ranking is fair if the utility profiles for users and for items are not jointly Lorenz-dominated:
\begin{definition}[Lorenz efficiency]
A utility profile $\tu\in\US$ is \emph{Lorenz-efficient} if there is no $\bv\in\US$ such that either ($\bvitems\lorenzw\buitems$ and $\bvusers\lorenz\buusers$) or ($\bvusers\lorenzw\buusers$ and $\bvitems\lorenz\buitems$).
\end{definition}
We consider that Lorenz-dominated profiles are undesirable (and unfair) because the utility of worse-off fractions of the population could have been increased at no cost for total utility. 
Examples of Lorenz-curves of users and items are given in Fig.~\ref{fig:tradeoffs}. The blue solid, green dotted and orange dashed curves are all non-dominated (the blue solid ranking has higher user utility but high item inequality, the green dotted and orange dashed curves have similar item exposure profiles, but user curves that intersect). On the other hand, the red dot/dashed curve is an unfair ranking: compared to the green dotted and orange dashed curve, all fractions of the worse off users have lower utility, together with less exposure for worse-off items.

A fundamental result from cardinal welfare economics is that concave welfare functions of utility profiles order profiles according to Lorenz dominance \citep{atkinson1970measurement,shorrocks1983ranking}. The choice of the welfare function specifies which (fair) trade-off is desirable in a specific context.  This result holds when all utilities are comparable. In our case where there are users and items,  we propose the following welfare
function parameterized by $\theta=(\lambda, \alpha_1, \alpha_2)$:%
\footnote{$\Wtheta(\bu)=-\infty$ if $\alpha\leq 0$ and $\exists i, \bu_i=0$. In practice, we use $\psi(x+\eta, \alpha)$ for $\eta>0$ to avoid this case.} 
\begin{align}
    \resizebox{\linewidth}{!}{$\displaystyle 
    \forall\bu\in\Re_+^n:~~
    \Wthetau = (1-\lambda)\sumU \psi(\tui, \alpha_1) 
    +\lambda\sumI \psi(\tuj, \alpha_2) 
    \text{~~with~}
    \psi(x, \alpha) = 
    \begin{cases} 
    x^\alpha & \text{~if~} \alpha >0\\
    \log(x) & \text{~if~} \alpha=0\\
    -x^{\alpha}& \text{~if~} \alpha <0
    \end{cases}.
    $}
\end{align}
Inference is carried out by maximizing $\Wtheta$ (an efficient algorithm is proposed in Section~\ref{sec:algo}):
\begin{align}\label{eq:inference}
    \text{\emph{(ranking procedure)}} && \rrkopt\in\argmaxPls~ \WthetauP\vphantom{\sum}
\end{align}
In $\Wtheta$, $\lambda\in[0,1]$ controls the relative weight of users and items. The motivation for the specific choice of $\psi$ is that it appears in scale invariant welfare functions \citep{moulin2003fair}, but other families can be used as long as the functions are \emph{increasing} and \emph{concave}. Monotonicity implies that maxima of $\Wtheta$ are Pareto-efficient. For $\alpha_1<1$ and $\alpha_2<1$, $\Wtheta$ is strictly concave. Then, $\Wtheta$ exhibits \emph{diminishing returns}, which is the key to Lorenz efficiency: an increment in utility for a worse-off user/item increases welfare more than the same increment for a better-off user/item. The effect of the parameters is shown in Fig.~\ref{fig:tradeoffs} (left): For \emph{item fairness} we obtain more item equality by using $\alpha_1<1$ (here, $\alpha_1=0.5$) and incrasing $\lambda$ (see blue solid vs orange dashed curve). The parameter $\alpha_2$ controls \emph{user fairness}: smaller values yield more user utility for the worse-off users at the expense of total utility, with similar item exposure curve (green dotted vs orange dahsed curves).
Let $\Theta = \xSet{(\lambda, \alpha_1, \alpha_2)\in\!(0,1)\times(\!{-}\infty,\!1)^2}$. For every $\theta\in\Theta, \Wtheta$ is strictly concave, and users and items have non-zero weight. We then have (the result is a straightforward consequence of diminishing returns, see Appendix~\ref{sec:proofswelfare}):
\begin{restatable}{proposition}{welfareisfair}\label{prop:efficiency}
$\forall\theta\in\Theta, \forall\rrkopt\in\argmaxPls \Wtheta(\tu(\rrk))$, $\rrkopt$ is Lorenz-efficient.
\end{restatable}

\paragraph{Relationship to inequality measures} A well-known measure of inequality is the Gini index, defined as $1 - 2\times {\rm AULC}$, where ${\rm AULC}$ is the area under the Lorenz curve. The difference between Lorenz and generalized Lorenz curves is that the former is normalized by the cumulative utility. This difference is fundamental: we can decrease inequalities while dragging everyone's utility to 0. However, this would lead to dominated \emph{generalized} Lorenz curves. Interestingly, for \emph{item-side} fairness, the cumulative exposure is a constant and thus trade-offs between user utility and  item exposure inequality are not really problematic. However, for user-side fairness, the total utility is not constant and reducing inequalities might require dragging the utility of some users down for the benefit of no one.

\paragraph{Additional theoretical results} 
In App.~\ref{sec:leximin}, we show that as $\alpha_1,\alpha_2 \rightarrow -\infty,$ utility profiles tend to leximin-optimal solutions \citep{moulin2003fair}.
Leximin optimality corresponds to increasing the utility of the worst-off users/items one a a time, similarly to a lexical order. 
In App. \ref{sec:regretbound}, we present an excess risk bound, which provides theoretical guarantees on the \emph{true} welfare when computing rankings based on \emph{estimated} preferences, depending on the quality of the estimates.

\subsection{Extension to reciprocal recommendation}\label{sec:reciprocal}

In reciprocal recommendation problems such as dating, the users are also items. The notion of fairness simplifies to increasing the utility of the worse-off users, which can in practice be done by boosting the exposure of worse-off users. Our framework above applies readily by taking $\userS=\itemS$ and $n=\card{\userS}$. The critical step however is to redefine the utility of a user to account for the fact that (1) the user utility comes from both the recommendation they receive and who they are recommended to, and (2) users have preferences over who they are recommended to. 

To define this \emph{two-sided utility}, let us denote by $\muij$ the mutual preference value between $i$ and $j$, and our examples follow the common assumption that $\muij=\muji$ (see e.g., \cite{palomares2021reciprocal}). For instance, when recommending CVs to recruiters, $\muij$ can be the probability of interview, while in dating, it can be that of a ``match''. The two-sided utility is then the sum of the user-side utility and item-sided utility of the user:
\begin{align}\label{eq:twosided}
    \overbrace{\vphantom{\sum}\UiP= 
    \sum_{j\in\itemS}
    \muij \rrkij\v}^{
    \substack{\text{user-side utility}\\\text{\emph{($j$ recommended to $i$)}}}
    }
    && \overbrace{\vphantom{\sum}\ViP=
    \sum_{j\in\userS}
    \muij\rrkji\v}^{
    \substack{\text{item-side utility}\\\text{\emph{($i$ recommended to $j$)}}}
    } 
    && \overbrace{\tuiP = \UiP + \ViP}^{\text{(two-sided) utility}} 
\end{align}
With this definition of two-sided utility, our previous framework can be readily applied using $\userS=\itemS$. A (two-sided) utility profile $\tu\in\US$ is \emph{Lorenz-efficient} if there is no $\bv\in\US$ such that $\bv\lorenz\bu$. The welfare function simplifies to $\Wthetau = \sum_{i=1}^n \psi(\bu_i, \alpha)$, and Proposition \ref{prop:efficiency} also holds true in this setting: maximizing the welfare function always yields Lorenz-efficient rankings. 

Fig.~\ref{fig:tradeoffs} (right) illustrates how decreasing $\alpha$ increases utilities for the worse-off users at the expense of total utility. It also shows a Lorenz-dominated (unfair) profile, in which all fractions from the worst-off to the better-off users have lower utility.

From now on, we refer to \emph{one-sided} recommendation for non-reciprocal recommendation.

\section{Comparison to utility/inequality trade-off approaches}\label{sec:fairness_exposure}

As stated in the introduction, leading approaches to fairness in ranking are based on utility/inequality trade-offs. We describe here the representative approaches we consider as baselines in our experiments. We then present theoretical results illustrating the undesirable behavior of some of them.

\subsection{Objective functions}

\paragraph{One-sided recommendation} In one-sided recommendation, the leading approach is to define exposure-based criteria for item fairness  \citep{singh2018fairness,biega2018equity}. The first criterion, \emph{equality of exposure}, aims at equalizing exposure across items. The second one, \emph{quality-weighted exposure}\footnote{We use here the terminology of \citep{wu2021tfrom}. This criterion has also been called ``disparate treatment'' \citep{singh2018fairness}, ``merit-based fairness'' \citep{singh2019policy} and ``equity of attention'' \citep{biega2018equity}.}, which is advocated by many authors, defines the \emph{quality} of an item as the sum of user values $\quaj=\sum_{i\in\userS} \muij$ and aims for item exposure proportional to quality. The motivation of quality-weighted exposure is to take user utilities into account in the extreme case where the constraint is strictly enforced. Interestingly, as we show later, this approach has bad properties in terms of trading off user and item utilities.

In our experiments, we use the standard deviation as a measure of inequality. Denoting by $\totexp=\nusers\norm{v}_1$ the total exposure and by $\totqua = \sum_{j\in\itemS}\quaj$ the total quality:
\begin{align*}
   \displaystyle
    \substack{\text{\emph{\normalsize quality-weighted }}\\\text{\normalsize\emph{exposure}}}
    ~~
    \Objgquabetau
    = \sum_{i\in\userS} \tui 
    -\beta \sqrt{\Penquau}
    \text{~~with~~} 
    \Penquau= \frac{1}{n}
    \sum\limits_{j\in\itemS} \Big( \tuj - 
    \frac{\quaj\totexp}{\totqua} \Big)^2
    .
\end{align*}
\begin{align}
    \hspace{-0.4cm}\substack{\text{\emph{\normalsize equality of }}\\\text{\normalsize\emph{exposure}}} && \Objutilbetau &= \sumU \bu_i 
    -\beta \sqrt{\Penu}& \text{with~~} 
    \Penu= 
    \sumI \frac{1}{n}\Big( \bu_j - 
    \frac{1}{\nitems} \sum_{j'\in\itemS}\bu_{j'} \Big)^2
    .\label{eq:parityutility}
\end{align}

Some authors use $\Penpu=\sum_{(j,j')\in\itemS^2} |\frac{\bu_j}{q_j} - \frac{\bu_{j'}}{q_{j'}}|$ instead of $\sqrt{\Penqua}$ \citep{singh2019policy,morik2020controlling,basu2020framework}. $\Penqua$ and $\Penp$ have qualitatively the same behavior. We propose $\Penquau$ as a computationally efficient alternative to $\Penp$, since it involves only a linear number of terms and $\sqrt{\Penqua}$ is convex and differentiable except on $0$. 

\paragraph{Reciprocal recommendation} For reciprocal recommendation, we consider as competing approach a trade-off between total (two-sided) utility and inequality of utilities, as measured by the standard deviation:
\begin{align}
    \hspace{-0.4cm}\substack{\text{\emph{\normalsize equality of }}\\\text{\normalsize\emph{utility}}} && \Objutilbetau &= \sumU \bu_i 
    -\beta \sqrt{\Penu}& \text{with~~} 
    \Penu= 
    \sumI \frac{1}{n}\Big( \bu_j - 
    \frac{1}{\nitems} \sum_{j'\in\itemS}\bu_{j'} \Big)^2
    .\label{eq:parityutility}
\end{align}

\subsection{Inequity and inefficiency of some of the previous approaches}

We point out here to two deficiencies of previous approaches. 

First, for one-sided recommendation, we show that in some cases, compared to the welfare approach with any choice of the parameter $\theta\in\Theta$, quality-weighted exposure leads to the undesirable behavior of \emph{decreasing user utility} while \emph{increasing inequalities of exposure} between items. This is formalized by the proposition below, which uses the following notation: for $\theta\in\Theta$, let $\tutheta = \argmaxUls \Wthetau$, and for $\beta>0$, let $\USgqua = \argmaxUls ~\Objgquabetau$.
\begin{restatable}{proposition}{qualityweighted}
\label{thm:qw_onesided}
The following claims hold irrespective of the choice of $\tugqua\in\USgqua$.

For every $d\in\mathbb{N}_*$ and every $N\in\mathbb{N}_*$, there is a one-sided recommendation problem, with $d+1$ items and $N(d+1)$ users, such that $\forall \theta\in\Theta$, we have:
\begin{align*}
    \big(\exists \beta>0, \tuthetausers \lorenz \tugquausers \text{~and~}
    \tuthetaitems\lorenz \tugquaitems\big)
    && \text{and} && 
    \lim_{\beta\rightarrow\infty} \frac{\sumU \tugquai}{\sumU\tuthetai} \xrightarrow[d\to\infty]{} \frac{5}{6}.
\end{align*}
\end{restatable}

Second, in reciprocal recommendation, striving for pure equality can even lead to $0$ utility for \emph{every user}, even that of the worst-off user. More precisely, we show that in some cases, compared to the welfare approach with any choice of parameter $\theta\in\Theta$, there exists $\beta>0$ such that equality of utility has lower utility for every user, eventually leading to 0 utility for everyone in the limit $\beta\rightarrow\infty$.

\begin{restatable}{proposition}{parityutility}\label{prop:parityutility}
For $\beta>0$, let $\USutil = \argmaxU \Objutilbetau$. The claim below holds irrespective of the choice of $\tuutil\in\USutil$. Let $n\geq 5$. There is a reciprocal recommendation task with $n$ users such that:
\begin{align}
    \forall \theta\in\Theta, \tutheta, \exists \beta>0:~~\forall i\in\intint{n}, \tuthetai> \tuutili
    &&
    \text{~and~}
    &&
    \smash{\lim_{\beta\rightarrow\infty}} \sumUls \tuutili = 0.
\end{align}
\end{restatable}

\paragraph{Proofs and additional results} All proofs are deferred to App.~\ref{sec:proofs_fairness_exp}, where we provide several additional results regarding the use of quality-weighted exposure and equality of exposure in reciprocal recommendation: We show in Prop.~\ref{thm:fe_reciprocal} that there are cases where both approaches lead to user utility profiles with Lorenz-dominated curves, and significantly lower total user utility than the welfare approach for any choice of the parameters. 

\section{Efficient inference of fair rankings with the Frank-Wolfe algorithm}\label{sec:algo}

 
 


We now present our inference algorithm for \eqref{eq:inference}.  Appendix~\ref{sec:fwproof} contains the proofs of this section and describes a similar approach for the objective functions of the previous section. From an abstract perspective, the goal is to find a maximum $\rrk^*$ such that:
\begin{equation}\label{eq:fwobj}
    \rrk^* \in \argmax_{\rrk\in\rrkS} \W(\rrk) \text{~~~with~~~} \W(\rrk) = \sum_{i=1}^n \Phii\!\bigg(\sum_{j=1}^n \muij(\rrkij+\rrkji)\v\bigg)
\end{equation}
where for every $i$, $\Phii:\Re_+\rightarrow\Re$ is concave increasing, $\muij\geq 0$ and $\v$ is a vector of non-negative non-increasing values.
Since $\W$ is concave and $\rrkS$ is defined by equality constraints, the problem above is a convex optimization problem. However, this is a global optimization problem over the rankings of all users, so a naive approach would require $\nusers\nitems^2$ parameters and $2\nusers\nitems$ linear constraints. The same problem arises with the penalties of previous work. In the literature, authors either considered applying the item-fairness constraints to each ranking individually \citep{singh2018fairness,basu2020framework}, which leads to inefficiencies with our definition of utility (see Appendix~\ref{sec:micro}), or resort to heuristics to compute the rankings one by one without guarantees on the trade-offs that are achieved \citep{morik2020controlling,biega2018equity}. 


Our approach is based on the Frank-Wolfe algorithm \citep{frank1956algorithm}, which was previously used in machine learning in e.g., structured output prediction or low-rank matrix completion \citep{jaggi2013revisiting}, but to the best of our knowledge not for ranking. Denoting $\xdp{X}{Y}=\sum_{ijk} X_{ijk} Y_{ijk}$ the dot product between tensors, the algorithm creates iterates $\rrk^{(t)}$ by first computing $\tilde{\rrk} = \argmax_{\rrk\in\rrkS} \xdp{\rrk}{\grad \W(\rrkt)}$ and then updating $\rrk^{(t)} = (1-\gamma^{(t)})\rrk^{(t-1)} + \gamma^{(t)}\tilde{\rrk}$ with $\gamma^{(t)} = \frac{2}{t+2}$ \citep{clarkson2010coresets}. Starting from an initial solution\footnote{In our experiments, we initialize with the utilitarian ranking (Proposition~\ref{prop:utilitarian}).}, the algorithm always stays in the feasible region without any additional projection step. 
%
Our main contribution of this section is to show that $\argmax_{\rrk\in\rrkS} \xdp{\rrk}{\grad \W(\rrkt)}$ can be computed efficiently, requiring only one sort operation per user after computing the utilities. In the result below, for a ranking tensor $\rrk$ and a user $i$, we denote by $\rkS(\rrki)$ the support of $\rrki$ in ranking space.%
\footnote{Formally, $\rkS(\rrki)=\xSet[\big]{\sigma:\itemS\rightarrow\intint{\nitems}\big | \,\sigma \text{~is one-to-one, and~}\forall j\in\itemS, \rrk_{ij\sigma(j)} > 0}$.} 
\begin{restatable}{theorem}{frankwolfemax}\label{lem:fwgrad}
Let $\tempmuij = \Phii'\big(\tu_i(\rrkt)\big) \muij + \Phij'\big(\tu_j(\rrkt)\big) \muji$. 
Let $\tilde{\rrk}$ such that: 

\myrescale{1}{\forall i\in\userS, \forall \tilde{\sigma}_i\in\rkS(\tilde{\rrki})$:~~ $\tilde{\sigma}_i(j)<\tilde{\sigma}_i(j') \implies \tempmu_{ij}\geq \tempmu_{ij'}}.
Then \myrescale{1}{\tilde{\rrk} \in \argmaxPls \xdp{\rrk}{\grad \W(\rrkt)}}.
\end{restatable}
Moreover, it produces a compact representation of the stochastic ranking as a weighted sum of permutation matrices. The number of iterations of the algorithm allows to control the trade-off between memory requirements and accuracy of the solution.
Using previous convergence results for the Frank-Wolfe algorithm \citep{clarkson2010coresets}, assuming each $\Phii''$ is bounded, we have:
\begin{restatable}{proposition}{fwcomplexity}\label{prop:fwcomplexity} Let $B=\max\limits_{i\in\intint{n}} \norm{\Phi_i''}_{\infty}$ and $U = \max\limits_{\bu\in\US} \norm{\bu}_{2}^2$. Let $K$ be the maximum index of a nonzero value in $\v$ (or $\nitems$). Then $\forall t\geq 1, \W(\rrk^{(t)}) \geq \max\limits_{\rrk\in\rrkS} \W(\rrk) - O(\frac{B U}{t})$. Moreover, for each user, an iteration costs $O(\nitems\ln K)$ operations and requires $O(K)$ additional bytes of storage.
\end{restatable}

\section{Experiments}\label{sec:experiments}\label{sec:experiments}
\begin{figure}[t]
    \centering
    \begin{subfigure}[b]{0.55\textwidth}
         \centering
         \includegraphics[width=\linewidth]{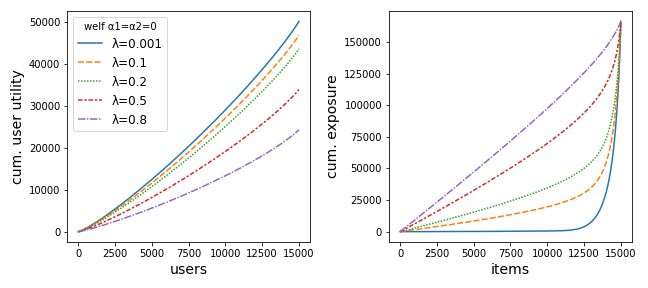}
         \caption{Examples of generalized Lorenz curves achieved by \welf.}
         \label{fig:lastfmsall_tradeoffs_welf}
     \end{subfigure}\begin{subfigure}[b]{0.245\textwidth}
         \centering
         \includegraphics[width=\linewidth]{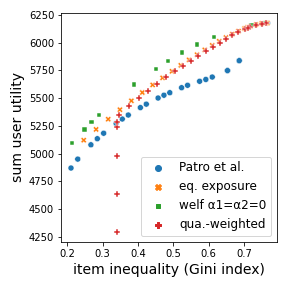}
         \caption{Summary of trade-offs}
         \label{fig:lastfmsmall_summary_tradeoffs}
     \end{subfigure}\begin{subfigure}[b]{0.20\textwidth}
         \centering
         \includegraphics[width=\linewidth]{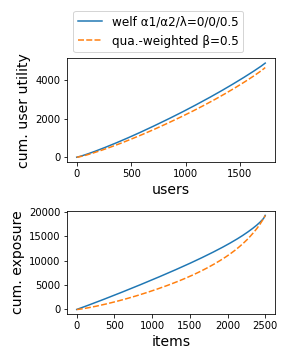}
         \caption{Dominated curve}
         \label{fig:lastfmsmall_qua_dominated}
     \end{subfigure}
    \caption{Summary of results on \lastfmsmall, focusing on the user utility/item inequality trade-off.}
    \label{fig:resultslastfm}
\end{figure}

\subsection{One-sided recommendation}\label{sec:experiments:1sided}

We first present experiments on movie recommendation task. We report here our experiments with the \lastfmsmall dataset \citep{Cantador:RecSys2011,patro2020fairrec}, which contains the music listening histories of $1.9k$ users. We present in App.~\ref{sec:lastfm360k} experiments on a larger portion of the Last.fm dataset, and in App.~\ref{sec:mlm} results using the MovieLens-20m dataset \cite{harper2015movielens}. Our results are qualitatively similar across the three datasets.

We select the top $2500$ items most listened to, and estimate preferences with a matrix factorization algorithm using a random sample of $80\%$ of the data. All experiments are carried out with three repetitions for this subsample. The details of the experimental protocol are in App.~\ref{app:experiments:1sided}. Since the goal is to analyze the behavior of the ranking algorithms rather than the quality of the preference estimates, we consider the estimated preferences as ground truth when computing user utilities and comparing methods, following previous work.
We compare our welfare approach (\welf) to three baselines. The first one is the algorithm of \citep{patro2020fairrec} (referred to as \patro in the figures), who consider envy-freeness for user-side fairness and, for item-side fairness, a constraint that the minimum exposure of an item is $\beta \frac{E}{\nitems}$ where $\beta$ is the trade-off parameter. The other baselines are quality-weighted exposure (\qualexpo) and equality of exposure (\equalexpo) as described in Sec.~\ref{sec:fairness_exposure}.


\paragraph{Item-side fairness} We first study in isolation item-side fairness, defined as improving the exposure of the worse-off item (producers). To summarize the trade-offs, we show the trade-offs by 
looking at exposure inequalities as measured by the Gini index (see Sec.~\ref{sec:welfare}). The results are given in Fig.~\ref{fig:resultslastfm}:
\begin{itemize}[leftmargin=*,noitemsep,nolistsep]
\item \emph{Generating user utility/item inequality trade-offs} is performed with our approach by keeping $\alpha_1=\alpha_2=0$ and varying the relative weight of items $\lambda$. Fig.~\ref{fig:lastfmsall_tradeoffs_welf} plots some trade-offs achieved by our approach.  As expected, the user utility curve degrades as we increase the weight of items, while at the same time the curve of item exposure moves towards the straight line, which corresponds to strict equality of exposure. Fig.~\ref{fig:appendix_lastfmsmall_all_trade_offs} in the appendix provides analogous curves for all methods, obtained by varying the weight $\beta$ of the inequality measure. 
\item {\it \qualexpo yields unfair trade-offs} Fig.~\ref{fig:lastfmsmall_qua_dominated} shows a \welf ranking that dominates a \qualexpo ranking on both user and item curves. This is in line with the discussion of Section~\ref{sec:fairness_exposure}, \qualexpo can lead to unfair rankings on utility/item inequality trade-offs. 
\item \emph{\welf dominates the user utility/item inequality (Gini) trade-offs} as seen on Fig. \ref{fig:lastfmsmall_summary_tradeoffs}: while all methods have the same total user utility when accepting high item inequality, \welf dominates \patro, \equalexpo and \qualexpo as soon as Gini $\leq 0.5$. Note, however, that the Gini index is only one measure of inequality. When measuring item inequalities with the standard deviation, \equalexpo becomes optimal since our implementation optimizes a trade-off with this measure (see Fig.~\ref{fig:appendix-lastfm-userfairness-std} in App.~\ref{app:experiments:1sided}). Overall, \welf and \equalexpo yield different fair trade-offs.
\end{itemize}

\paragraph{Two-sided fairness} Fig.~\ref{fig:lastfm-userfairness} shows the effect of the user curvature $\alpha_1\in\{-2,0,1\}$, keeping $\alpha_2=0$.  Fig.~\ref{fig:appendix-lastfm-userfairness-std} in App.~\ref{app:experiments:1sided} shows similar plots when the item inequality is measured by the standard deviation rather than the Gini index.
\begin{itemize}[leftmargin=*,noitemsep,nolistsep]
    \item \emph{Smaller $\alpha_1$ reduce user inequalities at the expense of total user utility, at various levels of item inequality}. This is observed by comparing the results for $\alpha_1\in\{-2,0,1\}$ in Fig.~\ref{fig:lastfmsmall_userfairness_giniutility} and Fig.~\ref{fig:lastfmsmall_userfairness_utility}.
    \item \emph{\welf $\alpha_1=0$ is better than \patro}, which can be seen by jointly looking at Fig.~\ref{fig:lastfmsmall_userfairness_utilitytenp}, \ref{fig:lastfmsmall_userfairness_utilitytwfp} and Fig.~\ref{fig:lastfmsmall_userfairness_utility} which give the cumulative utility at different points of the Lorenz curve ($10\%$, $25\%$ and $100\%$ of the users respectively). We observe that \welf $\alpha_1=0$ is similar to \patro at the $10\%$ and $25\%$ levels, but has higher total utility. Example curves are given in Fig.~\ref{fig:appendix_welf_nearlydom_patro_1} and \ref{fig:appendix_welf_nearlydom_patro_2} which plot \welf $\alpha_1=0$ and \patro at two levels of item inequality. \welf $\alpha_1=0$ obtains similar curves to \patro, except that it performs better at the end of the curve. A similar comparison can be made with \welf $\alpha_1=1$ and \equalexpo.
    \item \emph{More user inequalities is not necessarily unfair} as seen in Fig. \ref{fig:lastfmsmall_userfairness_giniutility} comparing \welf $\alpha_1=0$ and \patro. We observe that \welf $\alpha_1=0$ has slightly higher Gini index, but this is not unfair: as seen in Fig.~\ref{fig:appendix_welf_nearlydom_patro_1} and \ref{fig:appendix_welf_nearlydom_patro_2}, this is due to the higher utility at the end of the generalized Lorenz curve of \welf, but the worse-off users have similar utilities with \welf and \patro.
\end{itemize}

\begin{figure}[t]
    \centering
    \begin{subfigure}[b]{0.23\textwidth}
         \centering
          \includegraphics[width=\linewidth]{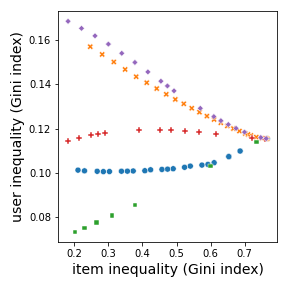}
         \caption{User inequality.~\\~}
         \label{fig:lastfmsmall_userfairness_giniutility}
     \end{subfigure}~~\begin{subfigure}[b]{0.23\textwidth}
         \centering
        \includegraphics[width=\linewidth]{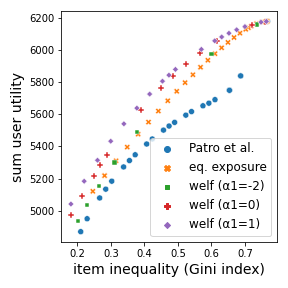}
         \caption{Total user utility.~\\~}
         \label{fig:lastfmsmall_userfairness_utility}
     \end{subfigure}~~\begin{subfigure}[b]{0.23\textwidth}
         \centering
         \includegraphics[width=\linewidth]{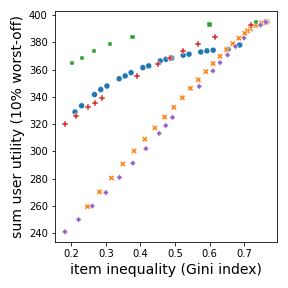}
         \caption{Cumulative utility at 10\% worse-off users.}
         \label{fig:lastfmsmall_userfairness_utilitytenp}
     \end{subfigure}~~
     \begin{subfigure}[b]{0.23\textwidth}
         \centering
         \includegraphics[width=\linewidth]{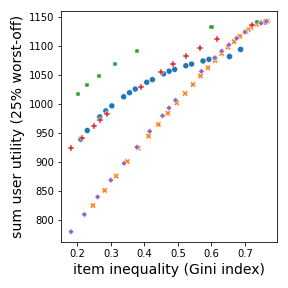}
         \caption{Cumulative utility at 25\% worse-off users.}
         \label{fig:lastfmsmall_userfairness_utilitytwfp}
     \end{subfigure}
     \begin{subfigure}[b]{0.45\linewidth}
    \centering
    \includegraphics[width=\linewidth]{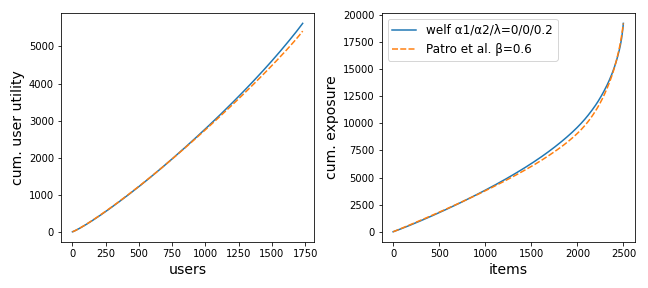}
    \caption{Item inequality $\approx 0.4$.}
    \label{fig:appendix_welf_nearlydom_patro_1}
    \end{subfigure}\quad
    \begin{subfigure}[b]{0.45\linewidth}
    \centering
    \includegraphics[width=\linewidth]{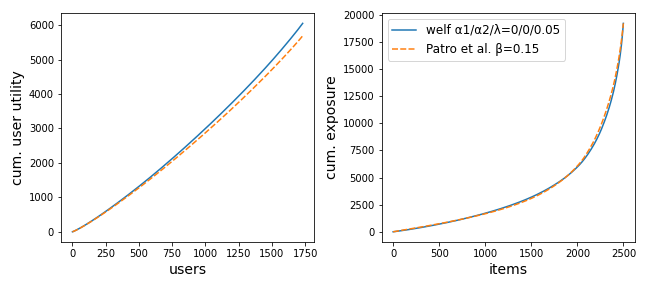}
    \caption{Item inequality $\approx 0.6$.}
    \label{fig:appendix_welf_nearlydom_patro_2}
    \end{subfigure}
    \caption{Summary of results on \lastfmsmall for two-sided fairness: effect of varying $\alpha_1$. 
    }
    \label{fig:lastfm-userfairness}
\end{figure}

\subsection{Reciprocal recommendation}\label{sec:experiments:reciprocal}

We now present results on a reciprocal recommendation task, where  fairness refers to increasing the utility of the worse-off users 
 (this can be done by boosting their exposure at the expense of total utility). 
Since there is no standard benchmark for reciprocal recommendation, we generate an artificial task based on the Higgs Twitter dataset \citep{de2013anatomy}, which contains follower links, and address the task of finding mutual followers (i.e., ``matches''). We keep users having at least $20$ mutual links, resulting in a subset of $13$k users. We build estimated match probabilities using matrix factorization. The experimental protocol is detailed in App.~\ref{app:experiments:reciprocal}. We also present in App.~\ref{sec:epinions} additional experiments using the Epinions dataset \cite{richardson2003trust}. The results are qualitatively similar.

Our main baseline is equal utility (\equalu) defined in Section~\ref{sec:fairness_exposure}.
We also compare to quality-weighted exposure, and equality of exposure as baselines that ignore the reciprocal nature of the task.
The results are summarized in Fig.~\ref{fig:twitter-uu}:
\begin{itemize}[leftmargin=*,noitemsep,nolistsep]
    \item \emph{Example of trade-offs obtained by varying $\alpha$} are plotted in Fig.~\ref{fig:twitter-uu}a. As $\alpha$ decreases, the utility increases for the worse-off users at the expense of better-off users. We note that increasing the utility of worse-off users has a massive cost on total user utility: looking at the exact numbers we observe that $\alpha=-5$ has more than doubled the cumulative utility of the 10\% worse off users compared to $\alpha=1$ (120 vs 280), but at the cost of more than 60\% of the total utility ($17$k vs $6.4$k). Fig.~\ref{fig:appendix_lastfmsmall_all_trade_offs} in Appendix~\ref{app:experiments:reciprocal} contains plots of the trade-offs achieved by the other methods.
    \item \emph{\qualexpo and \equalexpo are dominated by \welf on a large range of hyperparameters.} An example is given in Fig.~\ref{fig:twitter-uu}b, where \welf $\alpha=0.5$ already dominates some of their models, even though in this region of $\alpha$ there is little focus on worse-off users. More generally, all values of $\beta\geq 0.1$ for \qualexpo and \equalexpo lead to rankings with dominated curves. This is expected since they ignore the reciprocal nature of the task. 
    \item \emph{\equalu is dominated by \welf near strict equality} as illustrated in Fig.~\ref{fig:twitter-uu}c: for large values of $\beta$, it is not possible to increase the utility of the worse off users, and \equalu only drags utility of better-off users down.
    \item \emph{welf is more effective at increasing utility of the worse-off users} as can be seen in Fig.~\ref{fig:twitter-uu}e-g, which plots the total utility as a function of the cumulative utility at different points of the Lorenz curve ($10\%$, $20\%$, $50\%$ worse-off users respectively). For total utilities larger than $50\%$ of the maximum achievable, \welf significantly dominates \equalu in terms of utility of worse-off users (10\% and 25\%) at a given level of total utility. \welf also dominates \equalu on the $50\%$ worse-off users (Fig.~\ref{fig:twitter-uu}h) in the interesting region where the total utility is within 20\% of the maximum. 
    \item \emph{More inequality is not necessarily unfair} As shown in Fig.~\ref{fig:twitter-uu}d, we see that for the same utility for the $10\%$ worse-off users, \welf models have higher inequalities than \equalu. As seen before, this higher inequality is due to a higher total utility (and higher total utilities for the $25\%$ worse-off users. The analysis of these Lorenz curves allow us to conclude that these larger inequalities are not due to unfairness. They arise because \welf optimizes the utility of the worse-off users at lower cost in terms of average utility than \equalu.
\end{itemize}

\begin{figure}[t]
\centering
    \begin{subfigure}[b]{0.23\textwidth}
         \centering
         \includegraphics[width=\linewidth]{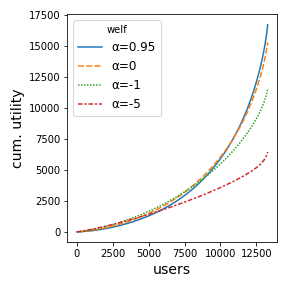}
         \caption{Example trade-offs achieved by \welf.}
         \label{fig:twitter-trafe-offs-welf}
     \end{subfigure}~~
     \begin{subfigure}[b]{0.23\textwidth}
         \centering
         \includegraphics[width=\linewidth]{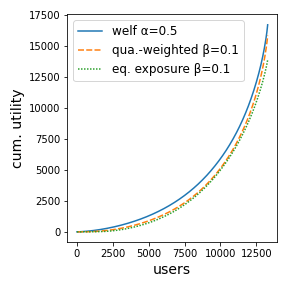}
         \caption{\welf dominates exposure-based methods.}
         \label{fig:twitter-welf-dom-quaexpo}
     \end{subfigure}~~
     \begin{subfigure}[b]{0.23\textwidth}
         \centering
         \includegraphics[width=\linewidth]{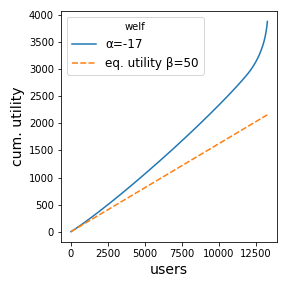}
         \caption{\welf dominates eq. utility near strict equality.}
         \label{fig:twitter-welf-dom-util}
     \end{subfigure}~~

    \begin{subfigure}[b]{0.23\textwidth}
         \centering
         \includegraphics[width=\linewidth]{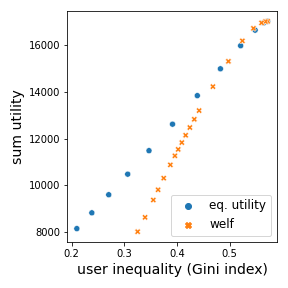}
         \caption{Total utility vs inequality.}
         \label{fig:twitter-uu-vs-ineq}
     \end{subfigure}~~
     \begin{subfigure}[b]{0.23\textwidth}
         \centering
         \includegraphics[width=\linewidth]{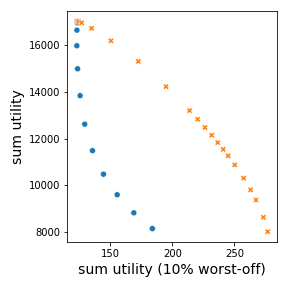}
         \caption{Total utility vs utility of 10\% worse-off users.}
         \label{fig:twitter-uu-first10}
     \end{subfigure}~~
     \begin{subfigure}[b]{0.23\textwidth}
         \centering
         \includegraphics[width=\linewidth]{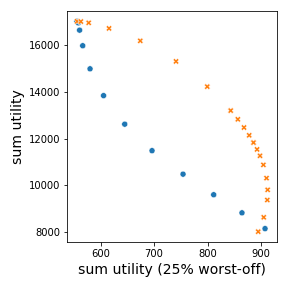}
         \caption{Total utility vs utility of 25\% worse-off users.}
         \label{fig:twitter-uu-first25}
     \end{subfigure}~~
     \begin{subfigure}[b]{0.23\textwidth}
         \centering
         \includegraphics[width=\linewidth]{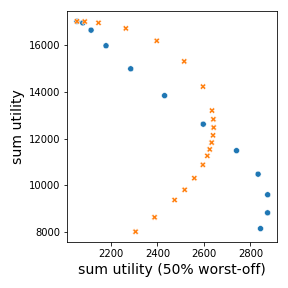}
         \caption{Total utility vs utility of 50\% worse-off users.}
         \label{fig:twitter-uu-first50}
     \end{subfigure}
    \caption{Results on the twitter dataset.}
    \label{fig:twitter-uu}
\end{figure}

\section{Related work}\label{sec:related}

The question of fairness in rankings originated from independent audits on recommender systems or search engines, which showed that results could exhibit bias against relevant social groups 
\citep{sweeney2013discrimination,kay2015unequal,hannak2014measuring,mehrotra2017auditing,lambrecht2019algorithmic}
Our work follows the subsequent work on ranking algorithms that promote fairness of exposure for individual or sensitive groups of items 
\citep{celis2017ranking,burke2017multisided,biega2018equity,singh2018fairness,morik2020controlling,zehlike2020reducing}. The goal is often to prevent winner-take-all effects, combat popularity bias \citep{abdollahpouri2019unfairness} or promote smaller producers \citep{liu2019personalized,mehrotra2018towards}. Section~\ref{sec:fairness_exposure} is devoted to the comparison with this type of approaches.
%
Most of these works use a notion of fairness oriented towards items only. Towards two-sided fairness, \citet{wang2020fairness} promote user-side fairness using concave functions of user utilities, similarly to us. Other works use equality constraints to define user-side fairness \citep{basu2020framework,wu2021tfrom}. These three approaches rely on the definitions of item-side fairness discussed in Section~\ref{sec:fairness_exposure}. \citet{patro2020fairrec} generate rankings that are envy-free on the user side, and guarantees the fair min-share for items. This approach is not amenable to controllable trade-offs between user and item utilities.

We are the first to address one-sided and reciprocal recommendation within the same framework. There is less existing work studying the fairness of rankings in the reciprocal setting. \citet{xia2019we} aim at equalizing user utility between groups, which suffers from the problems discussed in Section~\ref{sec:fairness_exposure}.
\citet{jia2018online} generate rankings using a welfare function approach, but optimizing only the utility of users \emph{being recommended}.
\citet{paraschakis2020matchmaking} postprocess rankings to correct for inconsistencies between estimated and declared preferences of users. 
In contrast, we aim at fair trade-offs between user and item utilities, under the assumption that biases in the preference estimates have been addressed earlier in the recommendation pipeline. Fairness is also studied in the context of ridesharing applications \citep{wolfson2017fairness,nixie2019balancing,nanda2020balancing}, but they address matching rather than ranking problems.

There is growing interest in making the relationship between fairness in machine learning and social choice theory \citep{heidari2018fairness,ustun2019fairness,balcan2018envy,golz2019paradoxes,hossain2020designing,chakraborty2019equality,DBLP:journals/corr/abs-2104-14527,finocchiaro2021bridging}, and welfare economics in particular \citep{speicher2018unified,hu2020fair,kleinberg2018algorithmic,lee2021formalising,zimmer2021learning}. In line with \citet{hu2020fair}, who focused on classification and parity penalties, we argue that Pareto-efficiency should be part of fairness assessments. We are the first to propose concave welfare functions and Lorenz dominance to address two-sided fairness in recommendation.

\section{Conclusion}\label{sec:discussion}

We view fairness in rankings as optimizing the distribution of user and item utilities, giving priority to the worse-off. Following this view, we defined fair rankings as having non-dominated generalized Lorenz curves of user and item utilities, and develop a new conceptual and algorithmic framework for fair ranking. The generality of the approach is showcased on several recommendation tasks, including reciprocal recommendation.


The expected positive societal impact of this work is to provide more principled approaches to mediating between several parties on a recommendation platform. Yet, we did not address several questions that are critical for the deployment of our approach. In particular, true user preferences are often not directly available, and we only observe proxies to them, such as clicks or likes. Second, interpersonal comparisons of utilities are critical in this work. It is thus necessary to make sure that the proxies we choose lead to meaningful comparisons of utilities between users. Third, estimating preferences or their proxies is itself not trivial in recommendation because of partial observability. The true fairness of our approach is bound to a careful analysis of (at least) these additional steps. 

\begin{ack}
We thank David Lopez-Paz, Jérôme Lang, as well as the anonymous NeurIPS reviewers for their constructive feedback on early versions of the paper. 
\end{ack}

\bibliography{references}
\bibliographystyle{abbrvnat}


\appendix
\newpage
\section{Outline of the appendices}

These appendices are structured as follows:
\begin{itemize}
    \item In Appendix~\ref{sec:groups}, we present how our fairness framework can be applied to sensitive groups of users or categories of items. 
    \item In Appendix~\ref{sec:proofswelfare}, we present a deeper analysis of the trade-offs achieved by the welfare approach. We also provide a theoretical guarantee relating the true welfare obtained by maximizing the welfare using estimated preferences, depending on the quality of the estimates.
    \item In Appendix~\ref{sec:proofs_fairness_exp}, we present the proofs for the theoretical results comparing our results and previous criteria of fairness in rankings. In addition, in Appendix~\ref{sec:parity_fairness_exposure}, we describe how to extend the criteria of equality of exposure and quality-weighted exposure in a reciprocal recommendation setting. This is the extension used in our experiments on reciprocal recommendation. In Proposition~\ref{thm:fe_reciprocal}, we present an additional result regarding the inefficiency of these criteria in reciprocal recommendation.
    \item In Appendix~\ref{sec:fwproof}, we present the more general version of the Frank-Wolfe algorithm, which we use both to optimize the welfare function over stochastic rankings, as well as the penalty-based baselines. This appendix also contains the proofs of the results in Section~\ref{sec:algo}. In addition, this appendix contains fundamental lemmas that are used in other appendices.
    \item Appendix~\ref{app:experiments} gives the details of the experiments presented in Section~\ref{sec:experiments}, as well we many additional experiments (two additional, larger scale datasets on one-sided recommendation, and an additional dataset for reciprocal recommendation)
    \item Appendix~\ref{sec:pairwisepointwise} briefly discusses the difference between the penalty we use in our implementation of the baseline approaches and an alternative penalty used by some authors.
    \item Finally, Appendix~\ref{sec:micro} discusses the difference between applying item-side fairness criteria for every ranking, compared to what we do in the paper, which defines item-side utility as an aggregate over the rankings of all users.
\end{itemize}
\section{Fairness towards sensitive groups rather than individuals}
\label{sec:groups}

\newcommand{\usergroups}{{\cal S}}
\newcommand{\nugroups}{{\card{\usergroups}}}
\newcommand{\ugp}{s_p}
\newcommand{\itemcats}{{\cal C}}
\newcommand{\nicats}{{\card{\itemcats}}}
\newcommand{\icq}{c_q}

In all the paper we focus on fairness towards individual users and items rather than groups of users or items. Prior work \citep{singh2018fairness,morik2020controlling,singh2019policy} considered the utlity of a group as the sum or the average utility of its members. Using this definition of group utility, our framework dirrectly extends to groups rather than individuals. In this section we describe the case of one-sided recommendation with groups of users and item categories. The case of reciprocal recommendation (with user groups only) is similar but simpler.

Let $\usergroups=(\ugp)_{p=1}^\nugroups$ be (possibly overlapping) user groups, i.e., $\forall p\in\intint{\nugroups}, \ugp\subseteq \userS$ and $\cup_{p\in\intint{\nugroups}} \ugp = \userS$. Similarly, let $\itemcats=(\icq)_{q=1}^\nicats$ be (possibly overlapping) item categories, i.e., $\forall q\in\intint{\nicats}, \icq\subseteq \itemS$ and $\cup_{q\in\intint{\nicats}} \icq = \itemS$. On the user side, such groups would typically correspond to demographic groups considered sensitive for the application at hand \citep{sweeney2013discrimination}. On the item side, groups can represent a single producer for the case where we want to be fair to producers based on the aggregate utility they obtain from their products \citep{mehrotra2018towards}, or demographic groups as well \citep{kay2015unequal}.

\newcommand{\grUp}{\tu_{\ugp}^{\rm gr}}
\newcommand{\grUpP}{\tu_{\ugp}^{\rm gr}(P)}
\newcommand{\grVq}{\tu_{\icq}^{\rm cat}}
\newcommand{\grVqP}{\tu_{\icq}^{\rm cat}(\rrk)}
\newcommand{\grU}{\boldsymbol{\tu}^{\rm gr}}
\newcommand{\grV}{\boldsymbol{\tu}^{\rm cat}}

In all cases, we redefine the user-side utility for groups and the item-side utility for categories:
\begin{align}
    \grUpP = \sum_{i\in\ugp} \tuiP&& \grVqP = \sum_{j\in\icq} \tujP
\end{align}
Let $\grU(\rrk) = (\grUpP)_{p=1}^\nugroups$ and $\grV(\rrk) = (\grVqP)_{q=1}^\nicats$ be the utility profiles of user groups and item categories associated to $\rrk$ respectively. The two-sided Lorenz efficiency for groups and categories is defined as:
\begin{definition}
Let $\usergroups$ be a set of user groups and $\itemcats$ a set of item categories. Let $\rrk\in\rrkS$. $\rrk$ is $(\usergroups, \itemcats)$-Lorenz efficient if there is no $\rrk'\in\rrkS$ such that either condition holds:
\begin{enumerate}
    \item $\grU(\rrk') \lorenzw \grU(\rrk)$ and $\grV(\rrk') \lorenz \grV(\rrk)$, or
    \item $\grV(\rrk') \lorenzw \grV(\rrk)$ and $\grU(\rrk') \lorenz \grU(\rrk)$.
\end{enumerate}
\end{definition}
The welfare function associated to $(\usergroups, \itemcats)$, still parametrized by $\theta=(\lambda, \alpha_1, \alpha_2)\in\Theta$, is defined as
\begin{equation}
    \Wtheta^{\rm gr}(\rrk) = (1-\lambda) \sum_{s\in\usergroups} \psi(\tu_{s}^{\rm gr}(\rrk), \alpha_1) + \lambda \sum_{c\in\itemcats} \psi(\tu_c^{\rm cat}(\rrk), \alpha_2)
\end{equation}
The welfare function follows the general form of objective function used for the algorithm in Appendix~\ref{sec:fwproof}, so the optimization of $\Wtheta^{\rm gr}$ requires similar computational complexity as $\Wtheta$.

Finally, the extension of Proposition~\ref{prop:efficiency} is straightforward. Its proof is similar to the proof presented in Appendix~\ref{sec:proofswelfare}.
\begin{proposition}
$\forall \theta\in\Theta, \forall\rrkopt\in\argmaxPl \Wtheta^{\rm gr}(\rrk)$, $\rrkopt$ is $(\usergroups, \itemcats)$-Lorenz efficient.
\end{proposition}

Note that this way of treating groups is not necessarily optimal. In particular, in does not account for within-group fairness. The separate consideration of within-group and between-group fairness has been studied extensively in the literature on equality of opportunity \citep{roemer2016equality}, which has inspired several works on algorithmic fairness \citep{hardt2016equality,heidari2019moral}. Yet, how to apply these principles to two-sided fairness in recommendation is still open, and is left as future work.
\section{More on welfare functions}\label{sec:proofswelfare}

This appendix provides an in-depth analysis of the trade-offs that are achievable by the welfare approach. We first pove the proposition of Section~\ref{sec:welfare}, and analyze the utilitarian rankings (obtained with $\alpha_1=\alpha_2=1$).  We then analyze how to obtain leximin optimal solutions on the side of the items in Appendix~\ref{sec:leximin}, as mentioned in Section~\ref{sec:welfare}. Finally, we prove Theorem~\ref{thm:regretbound} in Appendix~\ref{sec:regretbound}, which provides a regret bound relating the true welfare achieved when maximizing welfare on estimated preferences. Some results in this section use Lemma~\ref{lem:ranking} of Appendix~\ref{sec:algo}, which is proved in Appendix~\ref{sec:algo}.

Throughout the appendices, we use the more general version of item utilities (two-sided preferences), described at the end of Section~\ref{sec:framework}. Moreover, to clarify the notation, we remind that a \emph{ranking tensor} is a three-way tensor $\rrk$ where $\rrkijk$ is the probability that item $j$ is recommended to user $i$ at rank $k$. We consider $\rrk$ as an $\n\times\n\times\nitems$ tensor, where irrelevant entries are set to $0$. With this notation, the utility for both users and items can be written with the same formula:
\begin{equation}\label{eq:twosidedformula}
    \forall i\in\intint{n}, \tuiP = \sum_{j=1}^n \muij(\rrkij + \rrkji)\v.
\end{equation}
Note that this formula also corresponds to the two-sided utility in reciprocal recommendation. In general, the results in this appendix can be extended to reciprocal recommendation with minimal changes to their proofs, using $\userS=\itemS=\intint{n}$ and the formula above for the utility.

\subsection{Lorenz efficiency and utilitarian ranking}
We first prove Proposition~\ref{prop:efficiency}:
\welfareisfair*
\begin{proof}
It is well known that if $\Phi$ is increasing and strictly concave, then $F(\bu) = \sum_{i=1}^n \Phi(\bu_i)$ is monotonic with respect to Lorenz dominance \citep{shorrocks1983ranking,thistle1989ranking}: $\bu\lorenz\bv \implies F(\bu)>F(\bv)$. 

In the case of $\Wtheta$, for every $\theta=(\lambda, \alpha_1, \alpha_2)\in\Theta$, both $\psi(., \alpha_1)$ and $\psi(., \alpha_2)$ are strictly concave by the definition of $\Theta$ (recall that in $\Theta$, we have $\alpha_1, \alpha_2 < 1$). 

The partial function\footnote{We denote by $(\tuitems, {\tu}'_{\userS})$ the vector $\Re^d$ such that $(\tuitems, {\tu}'_{\userS})_i = \tu_i$ if $i\in\itemS$ and $(\tuitems, {\tu}'_{\userS})_i = u'_i$ if $i\in\userS$.} ${\tu}'_{\userS} \mapsto \Wtheta((\tuitems, {\tu}'_{\userS}))$ is, up to a constant, of the form of $F$ and likewise for the partial function ${\tu}'_{\itemS} \mapsto \Wtheta(({\tu}'_{\itemS},\tuusers))$. We now prove the result by contradiction. Assume that $\bu\in\argmaxU \Wthetau$ is not Lorenz efficient. Then there is $\bv\in\US$ such that $(\bvusers\lorenzw\buusers$ and $\bvitems\lorenz\buitems)$ or $(\bvusers\lorenz\buusers$ and $\bvitems\lorenzw\buitems)$. Let us assume $(\bvusers\lorenzw\buusers$ and $\bvitems\lorenz\buitems)$, the other case is dealt with similarly. We then have:
\begin{align}
    \Wtheta(\bv) &\geq \Wtheta((\bvitems, \buusers)) \tag*{(because $\bvusers\lorenzw\buusers$)}\\
    & > \Wtheta((\buitems, \buusers)) \tag*{(because $\bvitems\lorenz\buitems$)} 
\end{align}
which contradicts the maximality of $\bu$.
\end{proof}

The analogous for Proposition~\ref{prop:efficiency} for reciprocal recommendation is a direct consequence of standard results that concave welfare functions are monotonic with respect to Lorenz dominance \citep{shorrocks1983ranking,thistle1989ranking}.

\paragraph{Utilitarian ranking} Proposition~\ref{prop:utilitarian} below generalizes to two-sided utilities the well-known result that maximizing user-side utility is achieved by sorting $j\in\itemS$ by decreasing $\muij$ (see e.g., \citep{cossock2008statistical}). For a ranking tensor $\rrk$ and a user $i$, we denote by $\rkS(\rrki)$ the support of $\rrki$ in ranking space.%
\footnote{Formally, $\rkS(\rrki)=\xSet[\big]{\sigma:\itemS\rightarrow\intint{\nitems}\big | \,\sigma \text{~is one-to-one, and~}\forall j\in\itemS, \rrk_{ij\sigma(j)} > 0}$.} 
We remind that $\sigma(j)$ is the rank of item $j$, and that lower ranks are better. For a user $i$ and item $j$, we use $\mu_{ji} = 1$.
\begin{restatable}[Utilitarian ranking]{proposition}{proputilitarian} \label{prop:utilitarian}
Assume $\forall k\in\intint{\n-1}, \v_k>\v_{k+1}\geq 0$ and let
\begin{equation}
    \rrkopt\in\argmaxPls \W_{\frac{1}{2}, 1, 1}(\rrk) = \argmaxPls \sum\limits_{i\in\intint{\n}} \tuiP.
\end{equation}
\begin{enumerate}[leftmargin=*]
    \item $\displaystyle \forall i\in\userS, \forall \sigma\in\rkS(\rrkopti):
    \sigma(j) < \sigma(j') \implies \tmuij \geq \tmuijprime
    \text{~~with~~} \tmuij=\muij+\muji.$
    \item If $\forall (i,j)\in\intint{n}^2, \muij=\muji$, then $\tmuij\geq \tmuijprime\iff \muij\geq\mu_{ij'}$.
\end{enumerate}
\end{restatable}
When mutual preferences are symmetric (i.e., $\muij=\muji$), the utilitarian ranking is the same as the usual sort by decreasing $\muij$. This also obviously holds when we consider exposue as item utility ($\muji=1$). This means that without considerations of two-sided fairness ($\alpha_1, \alpha_2<1$), the optimal ranking for two-sided utilities is the same as the usual ranking. This might explain why the two-sided utility has never been studied before, even in reciprocal recommendation \citep{palomares2021reciprocal}.

For the proof of Proposition~\ref{prop:utilitarian}, the main part is the following lemma:
\begin{lemma}\label{lem:optiadd}
Let $F(\tuP) = \sum\limits_{i=1}^n \tuiP$ and $\tmuij = \muij + \muji$. Assume $\forall k\in\intint{\n-1}, \v_k\geq\v_{k+1}\geq 0$.

If $\rrk^*\in\rrkS$ is such that
$\forall \sigma \in \rkS(\rrki^*), \forall j, j'$,
%
$\sigma(j) < \sigma(j') \implies \tmuij \geq \tmuijprime$
then $\rrk^*\in\argmax\limits_{\rrk\in\rrkS} \tuP$. 

Moreover, if $\forall k\in\intint{\n-1}, \v_k>\v_{k+1}\geq 0$, then the reciprocal is true.
\end{lemma}
\begin{proof}
Notice that, thanks to the completion of $\rrk$ with zeros on irrelevant entries and formula \ref{eq:twosidedformula}, $F(\tuP)$ can be rewritten as:
\begin{align}
    F(\tuP) & = \sum\limits_{i=1}^n \tuiP = \sum_{i=1}^n \sum_{j=1}^n \muij(\rrkij + \rrkji)\v
     = \sum_{i=1}^n \sum_{j=1}^n (\muij+\muji)\rrkij \v
\end{align}
where the last equality is obtained by swapping $i$ and $j$ in the second sum, which is possible since $i$ and $j$ span the same range.

The result is then a direct consequence of Lemma~\ref{lem:ranking} in Appendix~\ref{sec:fwproof}, using $A_{ij}=\muij+\muji$.
\end{proof}

The first of statement of Proposition~\ref{prop:utilitarian} assumes that the exposure weights $\v$ are non-negative and strictly decreasing as per the second point of Lemma~\ref{lem:optiadd}. Lemma~\ref{lem:optiadd} above gives the statement for the more general case of non-increasing $\v$.
\begin{proof}[Proof of Proposition~\ref{prop:utilitarian}]
The first statement is the consequence of Lemma~\ref{lem:optiadd} above, noticing that $F(\tuP)$ in Lemma~\ref{lem:optiadd} always has the same argmax. The second statement is obvious from the assumptions. 

\end{proof}

\subsection{Item-side leximin optimality}\label{sec:leximin}

The most egalitarian trade-off achievable by our method is described by the leximin order \citep{sen1979equality}. Given two utility profiles $\bu$ and $\bv$,  $\bu\leximinw\bv$ if $\bucum$ is greater than $\bvcum$ according to the lexicographic order.\footnote{Formally, $\bu\leximin\bv$ if $(\exists k\in\intint{d}$ s.t. $\forall i< k$, $\bucum_i=\bvcum_i$ and $\bucum_k>\bvcum_k)$. ~~$\bu\leximinw\bv \iff \lnot(\bv\leximinw\bu)$.} The leximin optimal profile is egalitarian in the sense that it maximizes the utility of individuals in sequence, from the worse-off to the better-off. Depending on the set of feasible profiles, this may not lead to equal utility for everyone, but any further reduction of inequality can only be achieved by making people worse off for the benefit of no other, in violation of Pareto-dominance.

The proposition below formalizes how leximin optimal solutions on the side of items are found. It shows that item-side leximin solutions are obtained by having $\alpha_2\to-\infty$ and $\lambda\to 1$ at the same time. The proposition gives a formal statement of the rate at which $\lambda$ should converge to $1$ relative to $\alpha$.

In the statement of the proposition, given two functions $F$ and $G$, we use $F(\alpha) \geqinf G(\alpha)$ as a shorthand for $F(\alpha) \geq G(\alpha)$ for $\alpha$ sufficiently small.\footnote{Formally, $F(\alpha) \geqinf G(\alpha) \iff \exists\alpha_0\in\Re, \forall \alpha\leq \alpha_0, F(\alpha) \geq G(\alpha)$.}.

{
\newcommand{\USitemlex}{\US^{\rm item}_{\rm lex}}
\newcommand{\buopt}{\bu^*}
\newcommand{\buoptitems}{\bu^*_{\itemS}}

\begin{proposition}\label{prop:leximinitem} 
Let $\USitemlex = \xSet{\bu\in\US: \forall \bv\in\US, \buitems\leximinw\bvitems}$ and let $\buopt = 
\smash{
\argmax\limits_{\bu\in\USitemlex}} \sumUl \psi(\bu_i, \alpha_1)$.
%
\begin{align}
   \forall \eta>\max(1, \norm{{\buopt}_{\itemS}}_\infty), \forall \bu\in\US: \W_{1-\eta^\alpha, \alpha_1, \alpha}(\buopt) \geqinf \W_{1-\eta^\alpha, \alpha_1, \alpha}(\bu).
\end{align}
\end{proposition}
This means that among the leximin-optimal item-side utility profiles, $\alpha_1$ still controls the redistribution profile on the user side, since it is possible that $\card{\USitemlex}>1$ in one-sided recommendation. A similar result holds for user-side item leximin. 
\begin{proof}
Let $\buopt = 
\smash{
\argmax\limits_{\bu\in\USitemlex}} \sumUl \psi(\bu_i, \alpha_1)$ and $\bu\in\US$. Let $\theta=(\lambda, \alpha_1, \alpha)$ and take $\alpha<\min(0, \alpha_1)$.

\newcommand{\usjk}{u^*_{j_k}}
\newcommand{\ujk}{u_{j'_k}}
\newcommand{\usjm}{u^*_{j_m}}
\newcommand{\ujm}{u_{j'_m}}
\newcommand{\pa}[1]{\big(#1\big)^\alpha}

Let $(j_1, j_2, \ldots, j_{\nitems})$ be the ranking of $\buoptitems$ in increasing order: $u^*_{j_1} \leq \ldots u^*_{j_{\nitems}}$. Likewise, let $(j'_1, j'_2, \ldots, j'_{\nitems})$ be the ranking of $\buitems$ in increasing order: $u_{j'_1} \leq \ldots \leq u_{j'_{\nitems}}$.

Let $m = \max\{ k \in\intint{\nitems}\cup\{0\} : \forall \ell\leq k, u^*_{j_\ell} = u_{j'_\ell}\} + 1$, be the last index (+1) such that the smallest values of $\buopt$ and $\bu$ are equal ($m=1$ if the smallest values are different). 

Let $C(\alpha)=\W_{1-\eta^\alpha, \alpha_1, \alpha}(\buopt) - \W_{1-\eta^\alpha, \alpha_1, \alpha}(\bu)$.

Let $K = \sum_{i\in\userS} \big(\psi(\bu_i^*, \alpha_1) - \psi(\bv_i, \alpha_1)\big)$. 

\textbf{case 1: $m=\nitems+1$.} Then $C(\alpha) = (1-\eta^\alpha)K \geq 0$ since $\buoptitems = \buitems$ and $\buopt$ maximizes the user-side welfare.

\textbf{case 2: $m<\nitems$.} Then, we have $\ujm<\usjm$ by the leximin optimality of $\buoptitems$. We then have:
\begin{align}
    C(\alpha) &= (1-\eta^\alpha) K + \eta^\alpha\sum_{j\in\itemS} -(u^*_j)^\alpha + (u_j)^\alpha \\
    & = - (1-\eta^\alpha)\pa{\usjm} \Big(\frac{K}{1-\eta^\alpha} \underbrace{\pa{\frac{\eta}{\usjm}}}_{\xrightarrow[\alpha\to-\infty]{}0} + 1 + \sum_{k>m} \underbrace{\pa{\frac{\usjk}{\usjm}}}_{\xrightarrow[\alpha\to-\infty]{}0} - \underbrace{\pa{\frac{\ujm}{\usjm}}}_{\xrightarrow[\alpha\to-\infty]{}+\infty} - \sum_{k> m} \underbrace{\pa{\frac{\ujm}{\usjm}}}_{\geq 0}\Big)
\end{align}
which implies $\lim\limits_{\alpha\to-\infty}C(\alpha) =+\infty$ and thus the desired result.
\end{proof}
}

\subsection{Guarantees when performing inference with estimated preferences}\label{sec:regretbound}

In practice, inference is carried out on an estimate $\hmu$ of $\mu$, meaning that, denoting $\htu$ the resulting estimated utility\footnote{We have $\htu_i(\rrk)=\sumI \hmu_{ij}\rrkij\v$ for $i\in\userS.$} the system output $\hrrk = \argmaxP \Wtheta(\htu(\rrk))$. 
The following result extends surrogate regret bounds that exist in classification \citep{bartlett2006convexity,zhang2004statistical} and learning to rank \citep{cossock2008statistical,ravikumar2011ndcg,agarwal2014surrogate}
to the case of welfare functions and global stochastic rankings. It makes the link between the quality of the estimate $\hmu$ and an optimality guarantee for $\tu(\hrrk)$ (i.e., the true welfare of the ranking inferred on the estimated values). We prove the result for $\theta=(\frac{1}{2},\alpha,\alpha)$ for $\alpha\leq 1$ to simplify notation.\footnote{The dependency on $\hmu$ in $B(\hmu)$ is because $\psi'(., \alpha)$ is not bounded in general. In practice, we use $\psi(x+\eta, \alpha)$ for a small $\eta>0$ to avoid the singular point at $0$, in which case $B<\psi'(\eta, \alpha)$.}
\begin{restatable}{theorem}{regretbound}\label{thm:regretbound}
Let $\alpha\leq 1$ and $\theta=(\frac{1}{2}, \alpha, \alpha)\in\Theta$. Let $\hmu\in\Re_+^{\card{\userS}\times \card{\itemS}}$, $\hrrk = \argmaxP \Wtheta(\htu(\rrk))$, and $\rrkopt = \argmaxP \WthetauP$.

Let furthermore $B(\hmu) = \max\big(\max_{i\in\intint{n}} \psi'(\tui(\hrrk), \alpha), \max_{i\in\intint{n}} \psi'(\htui(\rrkopt), \alpha)\big)$.
We have:
\begin{align}
    \Wtheta(\tu(\rrkopt)) - \Wtheta(\tu(\hrrk)) \leq 4B(\hmu)\sqrt{n\norm{\v}_2^2}
    \sqrt{\sum_{(i,j)\in\userS\times\itemS}(\hmuij-\muij)^2}.
\end{align}
\end{restatable}

The existing results closest to our Theorem~\ref{thm:regretbound} are Theorem 2 of \citep{cossock2008statistical}. Here the result is substantially more difficult to prove because of the concave function and the fact that utilities are two-sided, calling for considering the rankings of multiple users at once.

\begin{proof}
We have:
\begin{align*}
    \Wtheta(\tu(\rrkopt)) - \Wtheta(\tu(\hrrk)) 
    &= \Wtheta(\tu(\rrkopt)) - \underbrace{\Wtheta(\htu(\hrrk))}_{\geq \Wtheta(\htu(\rrkopt)} + \Wtheta(\htu(\hrrk)) - \Wtheta(\tu(\hrrk))\\
    &\leq \underbrace{\Wtheta(\tu(\rrkopt)) - \Wtheta(\htu(\rrkopt))}_{=C_1} + 
    \underbrace{\Wtheta(\htu(\hrrk)) - \Wtheta(\tu(\hrrk))}_{=C_2}
\end{align*}

Let $B_1(\hmu) = \max_{i\in\intint{n}} \psi'(\htu_i(\rrkopt), \alpha)$. 

We first prove:
\begin{equation}\label{eqproof:boundA}
    C_1 \leq 2B_1(\hmu)\sqrt{n\norm{\v}_2^2}\sqrt{\sum_{(i,j)\in\intint{n}^2}(\hmu_{ij}-\muij)^2}.
\end{equation}
To prove \eqref{eqproof:boundA}, we start by using the concavity of $\psi(., \alpha)$ for $\alpha \leq 1$. Let $\Phi(.) = \frac{1}{2} \psi(., \alpha)$. We have:
\begin{align}
    C_1 & = \sum_{i=1}^n \big(\Phi(\tu(\rrkopt)) - \Phi(\htu(\rrkopt))\Big)
     \leq \sum_{i=1}^n \Phi'(\htu_i(\rrkopt))\big(\tu(\rrkopt)) - \htu(\rrkopt)\big)\\
    \text{thus~~} C_1& \leq \sum_{i=1}^n \sum_{j=1}^n \Phi'(\htu_i(\rrkopt))(\muij-\hmuij)(\rrkopt_{ij}+\rrkopt_{ji})\v\\
    & = \sum_{i=1}^n \sum_{j=1}^n \big(
    \underbrace{\Phi'(\htu_i(\rrkopt))(\muij-\hmuij) + \Phi'(\htu_j(\rrkopt))(\muji-\hmuji)}_{=A_{ij}} \big)\rrkopt_{ij}\v
\end{align}
where, similarly to the proof of Lemma~\ref{lem:optiadd}, we swapped the indexed $(i,j)$ in the $\Phi'(\htu_i(\rrkopt))\muij\rrkopt_{ji})\v$, which is possible because $i$ and $j$ span the same range in the sum.

Notice that the terms $A_{ij}\rrkopt_{ij}\v$ are all zero except if $i\in\userS$ and $j\in\itemS$ (because $\rrkopt_{ijk} = 0$ otherwise).
For $i\in\userS$, let $\sigma_i$ be a ranking which ranks $(A_{ij})_{j\in\itemS}$ in decreasing order, i.e., $\sigma_i(j)<\sigma_i(j') \implies A_{ij} \geq A_{ij'}$. Using Lemma~\ref{lem:ranking} in Appendix~\ref{sec:fwproof}, we have: 
\begin{align}
    C_1&\leq \max_{\rrk\in\rrkS} \sum_{i\in\userS}^n \sum_{j\in\itemS} A_{ij}\rrkij\v
     = \sum_{i\in\userS}^n \sum_{j\in\itemS} A_{ij}\v_{\sigma_i(j)}
\end{align}
Now let $V = [\v_{\sigma_i(j)}]_{\substack{i\in\userS\\j\in\itemS}}$. By Cauchy-Shwarz inequality and denoting $\norm{X}_F = \sqrt{\sum_{ij}X_{ij}^2}$ the Frobenius norm of matrix $X$, we have $\norm{V}_F = \sqrt{n\norm{\v}_2^2}$ and $\norm{A}_F \leq B_1(\hmu)(\norm{\mu-\hmu}_F + \norm{\mu^\top - \hmu^\top}_F)$, leading to:
\begin{align}
    C_1 \leq \sqrt{n\norm{\v}_2^2}\norm{A}_F \leq 2B_1(\hmu)\sqrt{n\norm{\v}_2^2}\norm{\mu-\hmu}_F
\end{align}
which proves \eqref{eqproof:boundA}.

Similarly, using $B_2(\hmu) = \max_{i\in\intint{n}} \psi'(\tui(\hrrk), \alpha)$ and the same arguments as above, we obtain:
\begin{equation}
    C_2 \leq \sqrt{n\norm{\v}_2^2}\norm{A}_F \leq 2B_2(\hmu)\sqrt{n\norm{\v}_2^2}\norm{\mu-\hmu}_F
\end{equation}
which yields the desired result.
\end{proof}
\section{Comparison to utility/inequality trade-offs}\label{sec:proofs_fairness_exp}
In this appendix, we provide the proofs of Section~\ref{sec:fairness_exposure}, and describe more precisely how we applied quality-weighted exposure and equality of exposure in reciprocal recommendation.
\subsection{One-sided recommendation: quality-weighted exposure}
We prove here Proposition~\ref{thm:qw_onesided} of Section~\ref{sec:fairness_exposure}. The result shows that in some cases, compared to any choice of the parameter $\theta\in\Theta$ of the welfare approach, quality-weighted exposure leads to the undesirable behavior of \emph{decreasing user utility} while \emph{increasing inequalities of exposure} between items. Figure~\ref{fig:fairnessexpo} gives an example. 
\qualityweighted*
\begin{proof}
We prove it for $N=1$, the more general case is just obtained by repeating the pattern with $d+1$ items and $d+1$ users.

Let $i_1, ..., i_{d+1}$ be the indexes of the users and $j_1, ..., j_{d+1}$ the indexes of the items. The preferences have the following pattern:
\begin{align}
    \forall k\in\intint{d+1}, \mu_{i_kj_k} = 1&& \forall k\in\intint{d}, \mu_{i_kj_{d+1}}=\frac{1}{2}
\end{align}
all other $\muij$ (for user $i$ and item $j$) are set to $0$ (note that we are in a problem with one-sided preferences, which means $\muji=1$ for every item $j$ and user $i$.

We consider a task with a single recommendation slot ($\v_1=1, \v_2 = \ldots=\v_{\nitems}=0$). On that problem, the optimal ranking for every $\theta \in \Theta$ is to show item $j_k$ to user $i_k$, which leads to perfect equality in terms of item exposure, and maximizes every user utility. It is thus leximin optimal for both users and items for every $\theta\in\Theta$.

Then, the qualities are equal to:
\begin{align}
    \forall k\in \intint{d}, q_{j_k} = 1 && q_{j_{d+1}} = \frac{1}{2}d + 1
\end{align}
the target exposure is thus $t_{j_k} = \frac{d+1}{\frac{3}{2}d+1}$ for $k\in\intint{d}$ and $t_{j_{d+1}} = (d+1)\frac{\frac{1}{2}d+1}{\frac{3}{2}d+1}$. 

Since the problem is symmetric in the users $i_1, ..., i_d$, by the concavity of $\Objgquabeta(\tuP)$ with respect to $\rrk$, there is an optimal ranking described by a single probability $p$ as:
\begin{align}
    \forall k\in\intint{d}, \rrk_{i_kj_k} &= 1-p &&  \rrk_{i_kj_{d+1}} = p&&
    \rrk_{i_{d+1}j_{d+1}} = 1
\end{align}
Note that for such a $\rrk$, $\forall k\in\intint{d}$, $u^{\text{qua}, \beta}_{i_k}(\rrk) = 1-\frac{1}{2}p$, and it is clearr that there is $\beta>0$ such that $p>0$, which then implies $\tutheta \lorenz \tugquausers \text{~and~}
\tutheta\lorenz \tugquaitems$.
    
Now, as $\beta\to\infty$, $p$ is such that exposure equals 
its target, which leads to the following equation:
\begin{align}
d p+1 = (d+1)\frac{\frac{1}{2}d+1}{\frac{3}{2}d+1}.
\end{align}
We thus get $p = \frac{d+1}{d} \frac{d+2}{3d+2} - \frac{1}{d}\xrightarrow[d\to\infty]{}\frac{1}{3}$, which gives the result $u^{\text{qua}, \beta}_{i_k}(\rrk) = 1-\frac{1}{2}p\xrightarrow[p\to\frac{1}{3}]{} \frac{5}{6}$.

Notice that similarly to Proposition~\ref{prop:parityutility}, the result does not depend on the choice of $\tugqua$ because the sum of user utilities converges.
\end{proof}

\subsection{Reciprocal recommendation: equality of exposure}

We now prove Proposition~\ref{prop:parityutility}.
\parityutility*
\begin{proof}
The example is given in Figure~\ref{fig:fairnessexpo}. We still consider a recommendation task with a single recommendation slot.

Let us rename the users by $i_1, i_2, ..., i_5$. The preference patterns are $\mu_{i_1i_2} = \mu_{i_1i_3} = 1$ and $\mu_{i_4i_5} = 1$. Apart from $\muij=\muji$, other $\muij$s are $0$. In this proof, we show that $\tu_{i_1}^{{\rm eq}, \beta} = 2\tu_{i_2}^{{\rm eq}, \beta}$ for every $\beta$, which implies that $\tu_{i_1}^{{\rm eq}, \beta} \xrightarrow[\beta\to\infty]{}0$ because $0$ utility for every user is feasible. On this task, the leximin ranking also maximizes the sum of users utilities (as shown in Figure~\ref{fig:fairnessexpo}), so the optimal ranking is the same for every $\theta\in\Theta$, and every user has a two-sided utility of at least $1.5$.

Since $\Objutilbetau$ is stricly Schur-concave for $\beta>0$, $i_2$ and $i_3$ always have the same utility in an optimal utility profile (because they play a symmetric role). $i_4$ and $i_5$ also have the same utility. Note that the interest of $i_4$ and $i_5$ in that problem is to make it possible to recommend them to $i_1$, which has $0$ value. 

Similarly to the problem in one-sided recommendation, the only way to decrease the penalty is to reduce the utility of $i_1, i_4, i_5$. However, reducing the utility of $i_1$ can only be done by either recommending $i_4$ or $i_5$ to $i_1$, or recommending $i_4/i_5$ to $i_2/i_3$. In all cases, decreasing $i_1$'s utility decreases $i_2/i_3$'s utilities.

More precisely, because of the symmetries and the concavity of $\ObjutilbetauP$ with respect to $\rrk$, for every $\beta>0$, there is an optimal ranking tensor described by three probabilities $p, q, q'$ such that:\footnote{Since there is a single recommendation slot, we identify $\rrk_{ij1}$ with $\rrkij$}
\begin{align}
    \rrk_{i_1i_2} &\!=\! \rrk_{i_1i_3} \!=\!  \frac{1}{2}p
    & \rrk_{i_2i_1} & \!=\! \rrk_{i_3i_1} \!=\! q
    & \rrk_{i_4i_5} &\!=\! \rrk_{i_5i_4} \!=\! q'
    \\
    \rrk_{i_1i_4} &\!=\! \rrk_{i_1i_5} \!=\! \frac{1}{2}(1-p)
    &  \rrk_{i_2i_3} &\!=\! \rrk_{i_2i_4} \!=\! \rrk_{i_2i_5} \!=\! \frac{1}{3}(1-q)
    & \rrk_{i_4i_1} &\!=\! \rrk_{i_4i_2} \!=\! \rrk_{i_4i_3} \!=\! \frac{1}{3}(1- q')
    \\
    && \rrk_{i_3i_2} &\!=\! \rrk_{i_3i_4} \!=\! \rrk_{i_3i_5} \!=\! \frac{1}{3}(1-q)
    & \rrk_{i_5i_1} &\!=\! \rrk_{i_5i_2} \!=\! \rrk_{i_5i_3} \!=\! \frac{1}{3}(1- q')
    \\
    %
    %
\end{align}
In all cases, the two-sided utility are
\begin{align}
    u_{i_1}(\rrk) &= \underbrace{p}_{\substack{\rrk_{i_1i_2}\mu_{i_1i_2}+\rrk_{i_1i_3}\mu_{i_1i_3}\\\text{user-side utility}}} + \underbrace{2q}_{\substack{\rrk_{i_2i_1}\mu_{i_2i_1}+\rrk_{i_3i_1}\mu_{i_3i_1} \\\text{item-side utility}}}&\text{~and~}&&
    u_{i_2}(\rrk) &= q+\frac{1}{2}p
\end{align}
Thus, in an optimal ranking for $\Objutilbetau$, we must have $u_{i_1}(\rrk) = 2u_{i_2}(\rrk)$. Equality, which is achieved at $\beta\to\infty$ can then only be at $0$ utility for every user (since $0$ is feasible).

The task used in the proof contains only $5$ users. Any number of users can be added to the group $\{i_4, i_5\}$, with a ``complete'' preference profile ($\muij=1$ for all pair $i,j$ in that group).
\end{proof}

The Lorenz efficiency of our welfare approach guarantees that it cannot exhibit the undesirable behaviors of equality or quality-weighted exposure penalties described in Propositions \ref{thm:qw_onesided} and \ref{thm:fe_reciprocal}.

\subsection{Equality of exposure and quality-weighted exposure in reciprocal recommendation}\label{sec:parity_fairness_exposure}

In one-sided recommendation with one-sided preferences, equality of exposure is the same as equality of utility. More generally, let $\gexpjP = \sum_{i\in\userS} \rrkij\v$ the total exposure of item $j$. Equality of exposure is defined by:
\begin{align}\label{eq:parity}
    \ObjgparbetaP &= \sum_{i\in\userS} \UiP - \beta  \sqrt{\sum\limits_{j\in\itemS} \Big( \gexpjP - 
    \frac{\nusers}{\nitems}\norm{\v}_1\Big)^2}
\end{align}
In one-sided recommendation, parity of exposure is relatively well behaved because the exposure target $\frac{\nusers}{\nitems}\norm{\v}_1$ is constant. Driving towards equality can thus not lead to a decrease of the total exposure budget, which was the problem with equality of utility in settings with two-sided preferences (driving towards equality of utility leads to a decrease of total utility), as we described in Section~\ref{sec:fairness_exposure}. 

The formula allows us to extend parity of exposure in the next section and in our experiments, since it is also valid in reciprocal recommendation. 
Likewise, the formula of quality-weighted exposure that is also valid in reciprocal recommendation is given by:
\begin{align*}
    \ObjgquabetaP
    = \sum_{i\in\userS} \UiP 
    -\beta \sqrt{
    \sum\limits_{j\in\itemS} \Big( \gexpjP - 
    \frac{\quaj\totexp}{\totqua} \Big)^2}
    .
\end{align*}

The result below shows that equality of exposure and quality-weighted exposure lead to inefficiencies in reciprocal recommendation settings:
\begin{proposition}\label{thm:fe_reciprocal}
For every $n\in\mathbb{N}_*$, there is a reciprocal recommendation task with $n$ users such that: 
\begin{align}
    \forall \theta\in\Theta, \exists\beta>0:&&\tutheta\lorenz\tugpar
    &&
    \text{~and~}
    &&
    \tutheta\lorenz\tugqua.
\end{align}
Moreover,  $\displaystyle\lim_{\beta\rightarrow\infty} \sum_{i\in\userS} \tugpari = \frac{2}{n} \sum_{i\in\userS} \tuthetai $ and $\displaystyle\lim_{\beta\rightarrow\infty} \sum_{i\in\userS} \tugquai = \frac{2+n}{2n} \sum_{i\in\userS} \tuuvi$.
\end{proposition}
\begin{proof}
An example of extreme case is with $n$ users when there is a ``leader'' who is the only possible match with other users. We consider a single recommendation slot. The preferences are:
\begin{align}
    \forall j\in\{2, \ldots, n\}, \mu_{1j} = \mu_{j1} =1&&  \forall (i,j)\in\{2, \ldots, n\}^2, \mu_{ij}=0.
\end{align}
On this task, the for every $\theta\in\Theta$, the optimal ranking is given by:
\begin{align}
    \forall j\in\{2, \ldots, n\}, \rrk_{1j} = \frac{1}{n-1} && \forall i\in\{2, \ldots, n\}  \rrk_{i1} = 1.
\end{align}
The reason it is the only possible optimal ranking is because it is leximin optimal and has the maximum achievable sum of utilities.
The utilities are then $\tu_1(\rrk) = n$ and $\tu_i(\rrk) = 1+\frac{1}{n-1}$, which leads to $\sum_{i=1}^n \tu_i = 2n$.

\textbf{Equality of exposure}
Driving towards equality of exposure requires to reduce the exposure of user $1$, which in turn reduces the utility of user $1$ and the utilities of those who user 1 is less exposed to. Thus, there is $\beta>0$ such that  $\tutheta\lorenz\tugpar$ because of the loss of efficiency. Finally, by the concavity of the objective with respect to $\rrk$, and by the symmetry of the problem with respect to $i_2, \ldots, i_n$, we can conclude that an optimal way to achieve perfect equality of exposure is to recommend, to every user $i$, every user $j\neq i$ with probability $\frac{1}{n-1}$. The utility is then $\tu_1(\rrk) = 1+ (n-1)\frac{1}{n-1}$ and $\tu_i(\rrk) = \frac{2}{n-1}$ for $i\geq 2$, leading to $\sum_{i=1}^n \tu_i = 4$, which gives the result.

\textbf{Quality-weighted exposure} On the same example, the qualities are $q_1 = n-1$ and $q_i = 1$ for $1\geq 2$. The total exposure targets are then $t_1 = \frac{1}{2}n$ and $t_i = \frac{n}{2(n-1)}$. These exposure targets mean less exposure for $1$ than in the leximin ranking. Thus $\beta$ sufficiently large has the effect of reducing $1$'s exposure\footnote{Direct calculations of the derivatives show that when $\beta>0$ is too small the penalty has no effect.}, which reduces the utility of $1$ and the users to whom $1$ is less recommended. Thus $\tutheta\lorenz\tugqua$. By the symmetry of the problem, as $\beta\to\infty$, quality weighted exposure is achieved by setting:
\begin{align}
    \forall j\in\{2,\ldots,n\}: \rrk_{1j} &= \frac{1}{n-1} & \rrk_{j1} = \frac{n}{2(n-1)} \\
    \forall j'\in\{2,\ldots,n\}, j'\neq j, \rrk_{jj'} &= \frac{1 - \frac{n}{2(n-1)}}{n-2} = \frac{1}{2(n-1)}
\end{align}
The utilities are then $\tu_1(\rrk) = 1+(n-1)\frac{n}{2(n-1)} = 1+\frac{n}{2}$ and $\tu_i(\rrk) = \frac{n}{2(n-1)} + \frac{1}{n-1} = \frac{n+2}{2(n-1)}$. The total utility is thus $2+n$, which gives the result.
\end{proof}
The Lorenz efficiency of our welfare approach guarantees that it cannot exhibit the undesirable behaviors of parity or quality-weighted exposure penalties described in Propositions \ref{thm:qw_onesided} and \ref{thm:fe_reciprocal}.

\begin{figure}
    \centering
    \includegraphics[height=3.11cm,trim={.22cm 3.3cm 6.4cm 0cm}, clip]{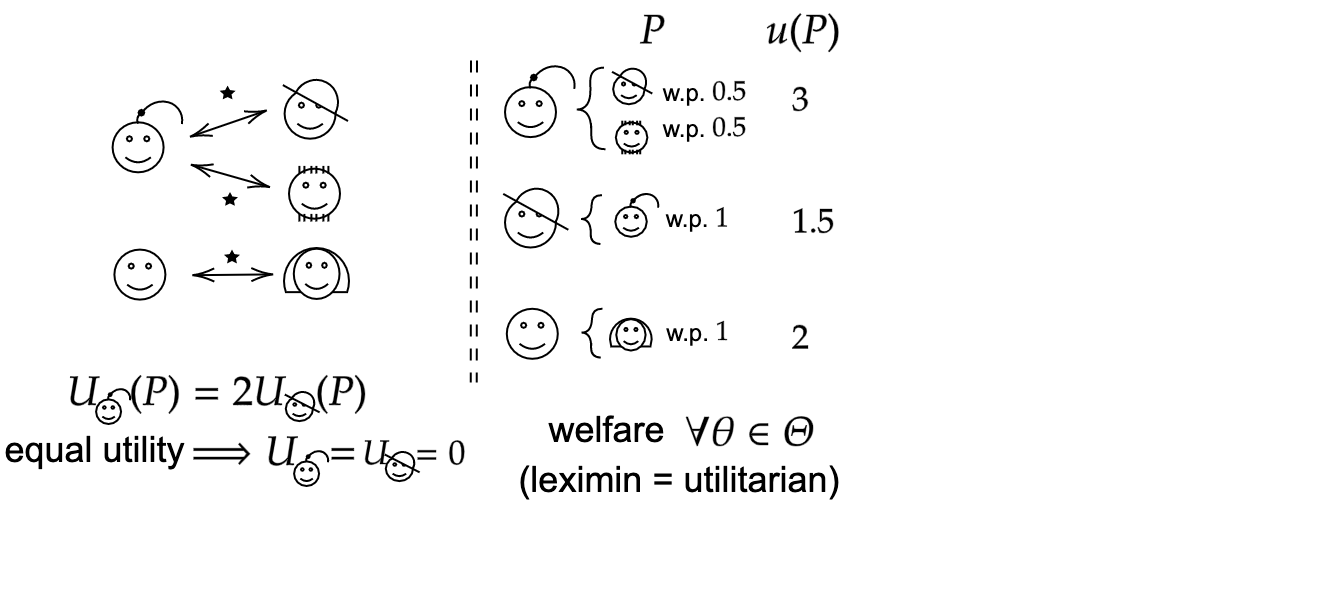}\quad\quad
    \includegraphics[height=3.11cm,trim={0 5.2cm 14.9cm 0cm},clip]{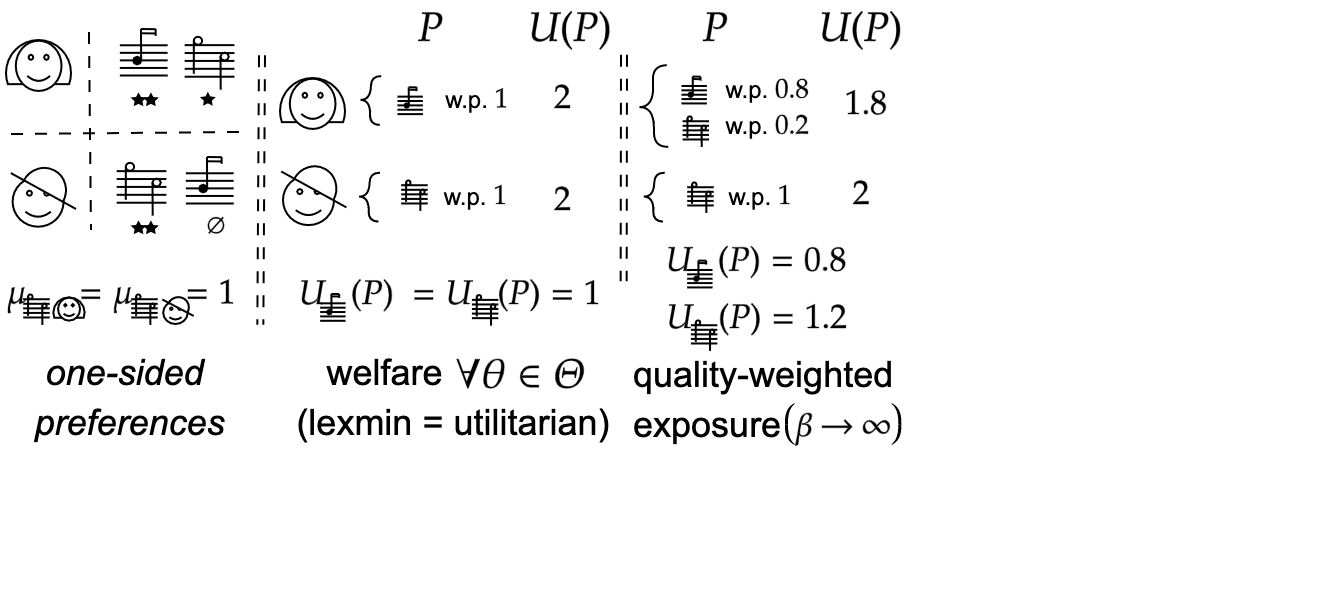}
    \caption{\textbf{Left:} Example of a reciprocal recommendation task where equality of utility leads to $0$ utility (see the proof of Prop. \ref{prop:parityutility} in App. \ref{sec:proofs_fairness_exp}). There is one recommendation slot per user. We give the recommendation probabilities and utilities for the utilitarian ranking and three users, the other ones are obtained by the symmetry of the problem. The utilitarian ranking is also leximin optimal, so our approach yields the same recommendations for all $\theta$. 
    \textbf{Right:} Example where quality-weighted exposure reduces user utility while increasing inequalities between items.}
    \label{fig:fairnessexpo}
\end{figure}

\section{A generic Frank-Wolfe algorithm for ranking}\label{sec:fwproof}

In this section, we present a general form of our algorithm presented in Section~\ref{sec:algo}, as well as the proofs of the claims.

Let $F:\Re^n\rightarrow \Re$, concave, and we want to find
\begin{equation}\label{eq:genealobjfw}
    \rrkopt\in\argmaxP F(\tuP).
\end{equation}

Let $\xdp{X}{Y} = \sum_{ijk} X_{ijk}Y_{ijk}$ be the dot product between three-way tensors, and let $\nabla (F\circ\tu)(\rrk)$ be the gradient of $\rrk\mapsto F(\tu(\rrk))$ taken at $\rrk$, i.e., $(\nabla (F\circ\tu))_{ijk} = \frac{\partial F\circ\tu}{\partial \rrkijk}$

Starting from $\rrk^{(0)}\in\rrkS$ (in our experiments we always use a utilitarian ranking $\rrk^{(0)}\in\argmaxP \sum_{i=1}^n \tuiP$), the Frank-Wolfe algorithm alternates two steps for $t\geq 1$:
\begin{enumerate}
    \item let $\tilde{\rrk} \in \argmaxP \xdp{\rrk}{\nabla (F\circ\tu)(\rrk^{t-1})}$
    \item $\rrk^{(t)} = (1-\gamma^{(t)}) \rrk^{(t-1)} + \gamma^{(t)} \tilde{\rrk}$ with $\gamma^{(t)} = \frac{2}{t+2}$
\end{enumerate}
The stepsize $\frac{2}{t+2}$ is from \citet[Section 3]{clarkson2010coresets}, which avoids a line search and in our experiments seemed to yield acceptable results. Irrespective of the step size, the fundamental results which allows to use Frank-Wolfe in the setting of \eqref{eq:genealobjfw} are the two following lemmas:
\begin{lemma}\label{lem:derivative}
Recall that $ \tuiP = \sum_{i=1}^n \mu_{ij}(\rrkij + \rrkji)\v$. Let $\frac{\partial F}{\partial \bu_i}$ denote the derivative of $F$ with respect to its $i$-th argument and $\frac{\partial F}{\partial \bu_i}(\tuP)$ the value of this derivative at $\tu(P)$. 

Then, $\forall i\in\userS, \forall j\in\itemS, \forall k\in \intint{\nitems}$, we have:
\begin{align}
\frac{\partial F\circ\tu}{\partial \rrkijk}(\rrk) = 
\Big( \muij\frac{\partial F}{\partial \bu_i} \big(\tuP\big) + \muji\frac{\partial F}{\partial\bu_j} \big(\tuP\big) \Big) \v_k.
\end{align}
\end{lemma}
\begin{proof}
The result is a consequence of the chain rule:
\begin{align}
    \frac{\partial F\circ\tu}{\partial \rrkijk}(\rrk) = \sum_{p=1}^n \frac{\partial F}{\partial \bu_p}(\tuP) \frac{\partial \tu_p(\rrk)}{\partial \rrkijk}
\end{align}
With 
\begin{equation}
    \tu_p(\rrk) = \sum_{q=1}^n \mu_{pq} \sum_{r=1}^{\nitems} (\rrk_{pqr}+\rrk_{qpr})\v_k.
\end{equation}
Thus $\frac{\partial \tu_p(\rrk)}{\partial \rrkijk} = (\mu_{ij}\indic{p=i} + \mu_{ji}\indic{p=j})\v_k$, which gives the desired result.
\end{proof}
\begin{lemma}\label{lem:ranking}
Let $A$ be an $\n\times\n$ matrix with $A_{ij}\in\Re$ (not necessarily non-negative). Let $\v\in\Re^{\nitems}$ with non-negative and non-increasing entries, i.e., $\forall k\in\intint{\nitems-1}$, $\v_k\geq \v_{k-1}\geq 0$. Let $K$ be the last index such that $\v_K>0$ (or $K=\nitems$ if there is no such index).

Let $\rrk\in\rrkS$ such that:
\begin{equation}
    \forall i, \forall \sigma_i\in\rkS(\rrki), \forall (j, j')\in\itemS^2: ~~\Big(\sigma_i(j) \leq K \text{~and~} \sigma_i(j) < \sigma_i(j') \implies A_{ij} \geq A_{ij'} \Big).
\end{equation} 
And let $X$ be the $\n\times\n\times\nitems$ tensor defined as $X_{ijk} = A_{ij}\v_k$.

Then $\rrk\in\argmaxP \xdp{\rrk}{X}$.

Moreover, if $\forall k\in\intint{\nitems-1}$, $\v_k> \v_{k-1}\geq 0$, then for every $\rrk\in\argmaxP \xdp{\rrk}{X}$, we have:
\begin{equation}
    \forall i, \forall \sigma_i\in\rkS(\rrki), \forall (j, j')\in\itemS^2: ~~\Big(\sigma_i(j) < \sigma_i(j') \implies A_{ij} \geq A_{ij'} \Big).
\end{equation}
\end{lemma}
\begin{proof}
The result stems from the rearrangement inequality (also known as the Hardy-Littlewood inequality \citep{hardy1952inequalities}), which states that for two vectors $a\in\Re_+^n$, and $b\in\Re^n$, $\argmax_{\nu} \sum_{j=1}^n a_{\nu(j)} b_j$, where $\nu$ spans the permutations of $\intint{n}$, is the set of permutations such that $b$ is ordered similarly to $(a_{\nu(i)})_{i=1}^n$. If the $a_k$s are non-increasing, then every permutation that sorts $b$ in decreasing order is in the argmax. We need the reciprocal statement for the second part of our Lemma: if the $a_i$s are strictly decreasing, then only the permutations that sort $b$ in decreasing order are in $\argmax_{\nu} \sum_{j=1}^n a_{\nu(j)} b_{j}$. Note that these arguments are well-known in learning to rank \citep[see, e.g., ][]{cossock2008statistical}.

In our case, notice that 
\begin{equation}
    \xdp{\rrk}{X} = \sum_{i\in\userS} \Big(\sum_{j\in\itemS} A_{ij}\rrkijk\v_k\Big)
\end{equation}
The maximization over $\rrk$ can then be performed over each user $i$ (and thus each bistochastic matrix $\rrki$ separately). Now, if $\rrki$ is such that every $\sigma_i\in\rkS(\rrki)$ orders $A_{ij}$ in decreasing order, then by the rearrangement inequality $\sigma_i\in\argmax_{\nu} \sum_{j\in\itemS} A_{ij} \v_{\nu(j)}$. Notice that if only the $K$ first elements of $\v$ are non-zero, we only need a top-$K$ ranking. This gives us the first part of the thoerem.

The second part of the theorem follows from the reciprocal of the rearrangement inequality, since for $\rrki$ to be an optimal stochastic ranking for $\sum_{j\in\itemS} A_{ij}\rrkijk\v_k$, every permutation $\sigma_i$ in its support must be in $\argmax_{\nu} \sum_{j\in\itemS} A_{ij} \v_{\nu(j)}$.
\end{proof}

\subsection{Proof of Theorem \ref{lem:fwgrad}}
Lemma~\ref{lem:derivative} and~\ref{lem:ranking} together are sufficient to give algorithms for the inference of stochastic rankings using our welfare function \eqref{eq:inference} and using the penalties of Section~\ref{sec:fairness_exposure}, by computing the partial derivatives $\frac{\partial F}{\partial \bu_i}$. The main result of Section~\ref{sec:algo}, which we prove now, instantiates this principle for the welfare function approach:
\frankwolfemax*
\begin{proof}
Notice that with $W(P)=F(\tuP)=\sum_{i=1}^n \Phii(\tuiP)$, then $\frac{\partial F}{\partial \bu_i}(\tuP) = \Phii'(\tuiP)$. By Lemma~\ref{lem:derivative}, we have that $\xdp{\rrk}{\nabla F(\rrk^{(t}})$ is of the form $\xdp{\rrk}{X}$ with $X_{ijk} = A_{ij}\v_k$ with $A_{ij} = \tempmuij$, so the result is implied by Lemma~\ref{lem:ranking}.
\end{proof}

\subsection{Proof of Proposition \ref{prop:fwcomplexity}}
\fwcomplexity*
\begin{proof}
Note that $\rrkS$ is a simplex over ranking tensors containing one deterministic ranking for each user. Using \citep[Section 3]{clarkson2010coresets}, the Frank-Wolfe algorithm with our step-size converges in $O\big(\frac{C_W}{t}\big)$, where, using \citep[Equation 11]{clarkson2010coresets} and denoting by $\nabla^2 W$ the Hessian of $W$, we have
\begin{align}
    C_W &\leq \sup\limits_{\substack{\bu,\bv\in\US\\ \tilde{\bf u}\in\US}} -\frac{1}{2}(\bu-\bv)^\top \nabla^2 W(\tilde{\bf u}) (\bu-\bv)
     \leq \frac{B}{2} \sup\limits_{\bu,\bv\in\US} \norm{\bu-\bv}^2_2 \leq 2B U.
\end{align}
where we used $\norm{\bu-\bv}_2^2\leq 2\norm{\bu}_2^2+2\norm{\bv}_2^2$.

For the computation cost, we use Lemma~\ref{lem:ranking}, which is more precise than Theorem \ref{lem:fwgrad}, to see that finding the argmax only requires a top-$K$ ranking. While technically any $\rrk\in\rrkS$ should contain a whole bistochastic matrix, it is not necessary to store a completion of the top-$K$ rankings because they have no impact on the utility. As such, storing each $\tilde{\rrk}$ only costs $O(K)$ bytes per user, which contain the indices of the top-K items in the ranking found by Theorem 2. 

Computing the two-sided utilities costs $O(\nusers\nitems)$, and thus $O(\nitems)$ per user. Moreover, computing the top-K ranking costs $O(\nitems\ln K)$ in the worst case, with a streaming method that maintains a min-heap of the top-K elements seen so far, and finish with sorting the top-$K$ elements. 
\end{proof}
Notice that for faster average performance, the top-K sort can be performed using a fast selection algorithm (such as quickselect), to obtain the top-$K$ elements with $O(\nitems)$ expected time complexity, and then sorting, yielding $O(\nitems+K\ln K)$ expected time complexity per user at each iteration. 
\newpage

\section{Additional experimental results}\label{app:experiments}

Our experiments are fully implemented in Python 3.9 using PyTorch\footnote{\url{http://pytorch.org}}. We provide the code as supplementary material. We compare our welfare maximization approach with the fairness penalties presented in Section \ref{sec:fairness_exposure}. 

We also compare ourselves to the algorithm FairRec from \citet{patro2020fairrec} (referred to as \patro in the figures and description), who consider envy-freeness as user-side fairness criterion, and max-min share of exposure as item-side fairness criterion. Envy-freeness states that every user should prefer their recommendation list to that of any other user. The max-min exposure criterion on the item side means that each user should receive an exposure of at least $\beta \frac{E}{\nitems},$ where $\beta$ is a parameter allowing to control how much exposure is guaranteed to items. We vary this parameter in our experiments to show the trade-offs achieved by \patro. Since \patro does not produce rankings, we took the recommendation list with the given order as a ranked list.

\subsection{One-sided recommendation: \lastfmsmall dataset}\label{app:experiments:1sided}

We describe in this section the details of the experiments presented in Section \ref{sec:experiments:1sided}. We use a dataset from the online music service Last.fm\footnote{\url{https://www.last.fm/}}. In the main paper, we presented results on \textbf{Lastfm-2k} from \citet{Cantador:RecSys2011} which contains real play counts of $2k$
users for $19k$ artists, and was used by \citet{patro2020fairrec} who also study two-sided fairness in recommendation. We filter the top $2,500$ items most listened to. Following \cite{johnson2014logistic}, we pre-process the raw counts with $\log$-transformation. We split the dataset into train/validation/test sets, each including $70\%/10\%/20\%$ of the user-item play counts. We create three different splits using three random seeds. One-sided preferences are estimated using the standard matrix factorization algorithm\footnote{Using the Python library Implicit: \url{https://github.com/benfred/implicit} (MIT License).} of \citet{hu2008collaborative} trained on the train set, with hyperparameters selected on the validation set by grid search. The number of latent factors is chosen in $[16, 32, 64, 128]$, the regularization in $[0.1, 1., 10., 20., 50.]$, and the confidence weighting parameter in $[0.1, 1., 10., 100.]$. The estimated preferences we use are the positive part of the resulting estimates.

Rankings are inferred from these estimated preferences. The exposure weights we use in the computation of utilities are the standard weights of the \emph{discounted cumulative gain} (DCG) (also used in e.g., \cite{singh2018fairness,biega2018equity,morik2020controlling}): $\forall k \in \intint{\nitems}, v_k = \frac{1}{\log_2(1+k)}$. For each ranking approach, the Frank-Wolfe algorithm is run with $5000$ iterations to make sure we are close to convergence, and the number of recommendation slots is set to $40$. 

We evaluate rankings on estimated preferences, considered as ground truth, following many works on fair recommendation \citep{singh2018fairness,patro2020fairrec,wang2020fairness,wu2021tfrom}. This is because the goal is to evaluate the fairness of ranking algorithms themselves, rather than biases in preference estimates. All results are averaged over three random seeds. 
To obtain various trade-offs, for \welf we vary $\lambda$ in $[0.001,0.01,0.05,0.075,0.1,0.125,0.15,0.2,0.3,0.325,0.35]$ and $[0.4,0.45,0.5,0.55,0.6,0.65,0.7,0.75,0.8,0.9,0.95,0.99,0.999]$. For \patro we vary $\beta$ in $[0.01,0.05,0.1,0.15,0.2,0.25,0.3,0.35,0.4]$ and $[0.45,0.5,0.55,0.6,0.65,0.7,0.75,0.8,0.85,0.9,0.95,1]$, and for other methods we vary $\beta$ in   $[0.001,0.005,0.01,0.015,0.0175,0.02,0.025,0.03,0.035,0.04,0.045,0.05,0.055,0.06]$ and $[0.065,0.07,0.075,0.08,0.085,0.09,0.095,0.1,0.105,0.11,0.2,0.5,1,2,5,10,20,30,40,50,70,100]$.

\paragraph{Item-side fairness} Figure \ref{fig:appendix_lastfmsmall_all_trade_offs} presents the various trade-offs achieved by each method in one-sided recommendation, as discussed in Section \ref{sec:experiments:1sided}. We observe that only \qualexpo is unable to reach equal exposure because of its quality-weighted exposure target: perfectly equal exposure is only permitted when all items have the same quality.

\paragraph{Two-sided fairness} Figure \ref{fig:appendix-lastfm-userfairness} shows the effect of varying $\alpha_1$ and $\lambda$ on user fairness as in Figure \ref{fig:lastfm-userfairness} of the main paper, but with results repeated over three random seeds. We observe the same trade-offs and conclude again that \welf is better than \patro and \equalexpo, in terms of its impact on worse-off users.

\paragraph{The importance of considering the whole Lorenz curve} In Fig.~\ref{fig:appendix-lastfm-userfairness-std} we show the results of the same models as before, but changing the way we measure the item inequality: using the standard deviation of exposure rather than the Gini index. Now, \equalexpo dominates the total utility/item inequality plot, since the plot corresponds exactly to the objective function of the algorithm. Comparing \equalexpo with \welf $\alpha_1=1$, we now see that the trade-offs are different, with \equalexpo performing better on the worse-off users. Comparing \welf $\alpha_1=0$ and \patro, we see that they still exhibit similar behaviors, with \welf $\alpha_1=0$ being better for better off users. Finally, \welf $\alpha_1=-2$ still dominates the othe methods in terms of performance on the worse-off users.

\begin{figure}
    \centering
    \begin{subfigure}[b]{0.23\linewidth}
    \centering
    \includegraphics[width=\linewidth]{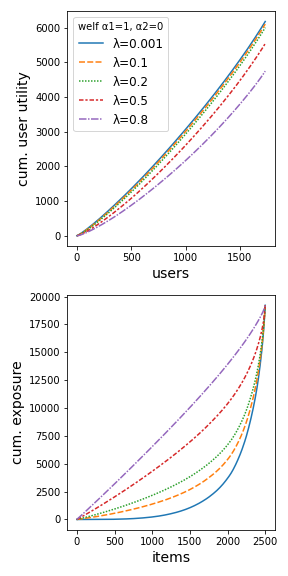}
    \caption{Welfare (ours)}
    \label{fig:my_label}
    \end{subfigure}
    \begin{subfigure}[b]{0.23\linewidth}
    \centering
    \includegraphics[width=\linewidth]{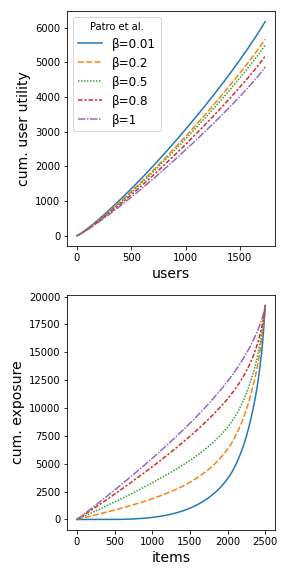}
    \caption{\patro}
    \label{fig:my_label}
    \end{subfigure}
    \begin{subfigure}[b]{0.23\linewidth}
    \centering
    \includegraphics[width=\linewidth]{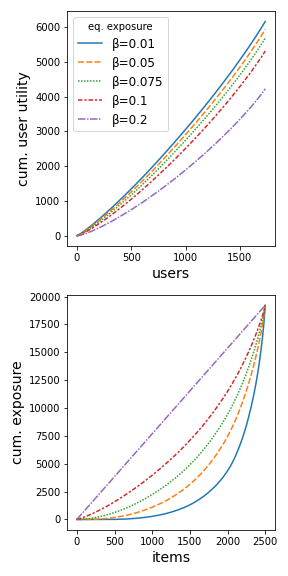}
    \caption{\equalexpo}
    \label{fig:my_label}
    \end{subfigure}
    \begin{subfigure}[b]{0.23\linewidth}
    \centering
    \includegraphics[width=\linewidth]{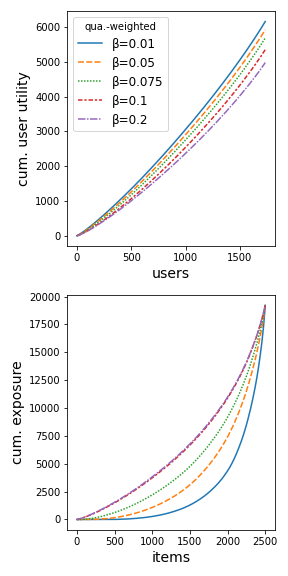}
    \caption{\qualexpo}
    \label{fig:my_label}
    \end{subfigure}
    \caption{representative trade-offs achieved by the various compared methods on \lastfmsmall. The trade-offs achieved by the different methods look alike, except that \qualexpo does not aim at reaching equality of exposure for exteme values of $\beta$. See Section~\ref{sec:experiments:1sided} for the discussions on the differences between the trade-offs achieved by the different approaches.}
    \label{fig:appendix_lastfmsmall_all_trade_offs}
\end{figure}

\begin{figure}
    \centering
    \begin{subfigure}[b]{0.23\textwidth}
         \centering
         \includegraphics[width=\linewidth]{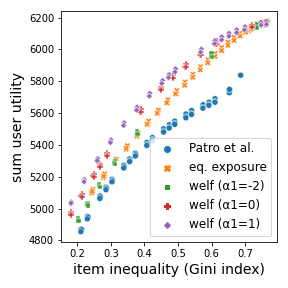}
         \caption{Total user utility.\\~}
         \label{fig:appendix-lastfmsmall_userfairness_utility}
     \end{subfigure}\begin{subfigure}[b]{0.23\textwidth}
         \centering
         \includegraphics[width=\linewidth]{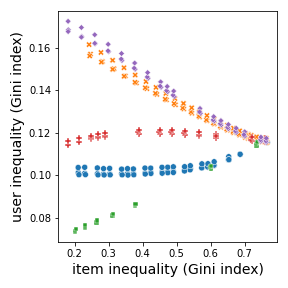}
         \caption{User inequality.\\~}
         \label{fig:appendix-lastfmsmall_userfairness_giniutility}
     \end{subfigure}\begin{subfigure}[b]{0.23\textwidth}
         \centering
         \includegraphics[width=\linewidth]{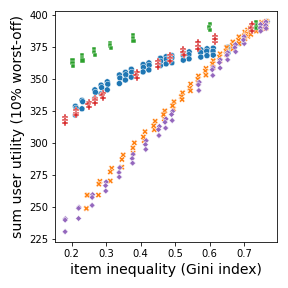}
         \caption{Cumulative utility for 10\% worse-off users}
         \label{fig:appendix-lastfmsmall_userfairness_utilitytenp}
     \end{subfigure}
     \begin{subfigure}[b]{0.23\textwidth}
         \centering
         \includegraphics[width=\linewidth]{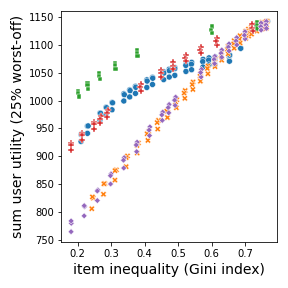}
         \caption{Cumulative utility for 25\% worse-off users}
         \label{fig:appendix-lastfmsmall_userfairness_utilitytwfp}
     \end{subfigure}
    \caption{Focus on user fairness on \lastfmsmall: effect of varying $\alpha_1$ (user-side curvature of the welfare function) keeping $\alpha_2=0$. The figure shows all the results obtained with a repetition of three seeds. Overlapping points correspond to the same model parameter across different seeds. We can see that the variance is negligible compared to the observed differences.}
    \label{fig:appendix-lastfm-userfairness}
\end{figure}

\begin{figure}
    \centering
    \begin{subfigure}[b]{0.23\textwidth}
         \centering
         \includegraphics[width=\linewidth]{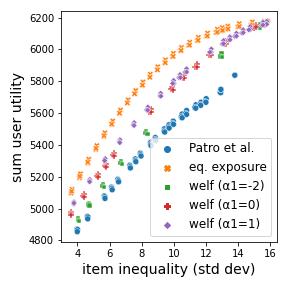}
         \caption{Total user utility.\\~}
         \label{fig:appendix-lastfmsmall_userfairness_utility}
     \end{subfigure}\begin{subfigure}[b]{0.23\textwidth}
         \centering
         \includegraphics[width=\linewidth]{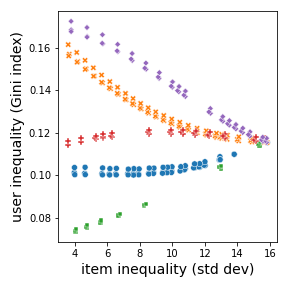}
         \caption{User inequality.\\~}
         \label{fig:appendix-lastfmsmall_userfairness_giniutility}
     \end{subfigure}\begin{subfigure}[b]{0.23\textwidth}
         \centering
         \includegraphics[width=\linewidth]{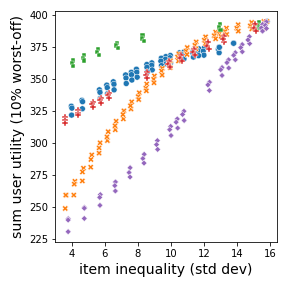}
         \caption{Cumulative utility for 10\% worse-off users}
         \label{fig:appendix-lastfmsmall_userfairness_utilitytenp}
     \end{subfigure}
     \begin{subfigure}[b]{0.23\textwidth}
         \centering
         \includegraphics[width=\linewidth]{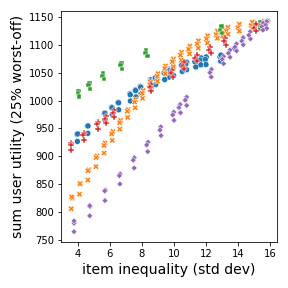}
         \caption{Cumulative utility for 25\% worse-off users}
         \label{fig:appendix-lastfmsmall_userfairness_utilitytwfp}
     \end{subfigure}
    \caption{Focus on user fairness on \lastfmsmall, measuring item inequality with standard deviation rather than Gini index. We observe a similar relative behavior between \welf and \patro, but now equality of exposure is optimal on the total utility/item inequality trade-off since it corresponds exactly to the objective of the algorithm. Nonetheless, \welf $\alpha_1=-2$ still obtains higher performance on 10\%-25\% worse-off users, showing that \welf offers a larger range of trade-offs than \equalexpo.}
    \label{fig:appendix-lastfm-userfairness-std}
\end{figure}

\subsection{One-sided recommendation: \lastfm dataset}\label{sec:lastfm360k}

We replicate the experiments on a larger dataset to verify our conclusions at a larger scale. We consider another Lastfm dataset from \citet{Celma:Springer2010}, which includes $360k$ users and $180k$ items (artists). We select the top $15,000$ users and items having the most interactions, so we refer to this dataset as \lastfm. We apply exactly the same experimental protocol as for \lastfmsmall, with the same range of hyperparameters for the different methods.

\paragraph{Results} Fig.~\ref{fig:appendix-lastfm360k-userfairness} and \ref{fig:appendix-lastfm360k-userfairness-std} show the results obtained by \welf, \patro and \equalexpo. The conclusions are similar to those on \lastfmsmall, with the results of \welf $\alpha=0$ being more uniformly better than those of \patro, even though overall similar. \welf $\alpha_1=-2$ dominates in terms of user utility on worse-off users. \welf and \equalexpo still find different trade-offs, with \welf dominating \equalexpo when inequality between items is measured by the Gini index, and \equalexpo dominating \welf when inequality is measured by the standard deviation.

\begin{figure}
    \centering
    \begin{subfigure}[b]{0.23\textwidth}
         \centering
         \includegraphics[width=\linewidth]{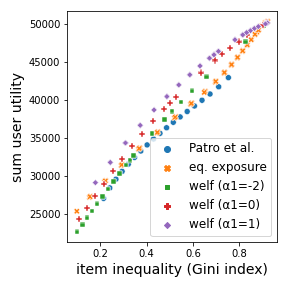}
         \caption{Total user utility.\\~}
     \end{subfigure}\begin{subfigure}[b]{0.23\textwidth}
         \centering
         \includegraphics[width=\linewidth]{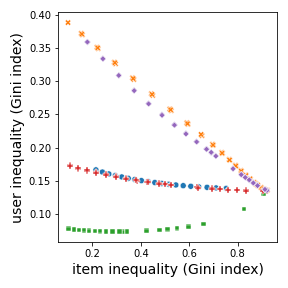}
         \caption{User inequality.\\~}
     \end{subfigure}\begin{subfigure}[b]{0.23\textwidth}
         \centering
         \includegraphics[width=\linewidth]{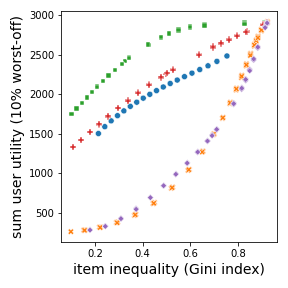}
         \caption{Cumulative utility for 10\% worse-off users}
     \end{subfigure}
     \begin{subfigure}[b]{0.23\textwidth}
         \centering
         \includegraphics[width=\linewidth]{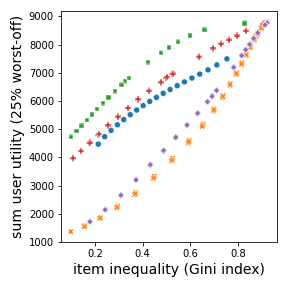}
         \caption{Cumulative utility for 25\% worse-off users}
     \end{subfigure}
    \caption{Results on \lastfm when measuring the inequality between items with the Gini coefficient.}
    \label{fig:appendix-lastfm360k-userfairness}
\end{figure}

\begin{figure}
    \centering
    \begin{subfigure}[b]{0.23\textwidth}
         \centering
         \includegraphics[width=\linewidth]{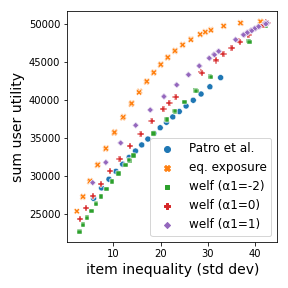}
         \caption{Total user utility.\\~}
     \end{subfigure}\begin{subfigure}[b]{0.23\textwidth}
         \centering
         \includegraphics[width=\linewidth]{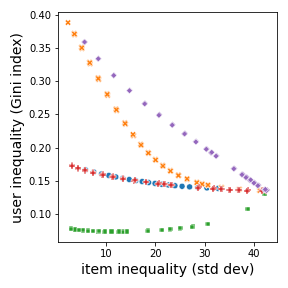}
         \caption{User inequality.\\~}
     \end{subfigure}\begin{subfigure}[b]{0.23\textwidth}
         \centering
         \includegraphics[width=\linewidth]{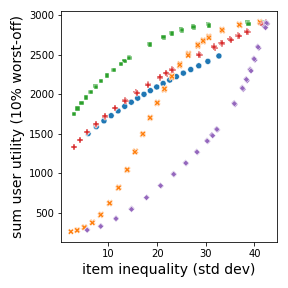}
         \caption{Cumulative utility for 10\% worse-off users}
     \end{subfigure}
     \begin{subfigure}[b]{0.23\textwidth}
         \centering
         \includegraphics[width=\linewidth]{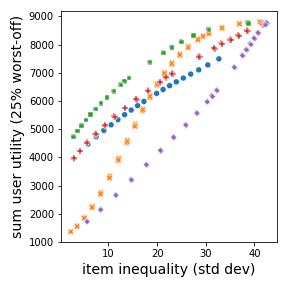}
         \caption{Cumulative utility for 25\% worse-off users}
     \end{subfigure}
    \caption{Results on \lastfm when measuring inequalities between  items with the standard deviation.}
    \label{fig:appendix-lastfm360k-userfairness-std}
\end{figure}

\subsection{One-sided recommendation: \mlm dataset}\label{sec:mlm}

We provide additional results on the \textbf{MovieLens}-20m dataset \cite{harper2015movielens}, which
contains ratings on a 5-star scale of movies by real users. 
To simulate a collaborative filtering task with implicit feedback similar to Last.fm, we consider missing ratings as negative feedback and the task is to predict positive values. Since ratings $<3$ are usually considered as negative \cite{lim2015top,wang2018modeling}, we set ratings $<3$ to zero, resulting in a dataset with preference values among $\{0,3,3.5,4,4.5,5\}$. 
As for Lastfm-15k, we select the top $15,000$ users and items with the most interactions. For the inference and evaluation of rankings, we follow the same protocols as for Last.fm. 

The experimental protocol is the same as for \lastfmsmall and \lastfm except that we do not run the algorithm by \citep{patro2020fairrec} because its runtime was prohibitive.

\paragraph{results} The results are qualitatviely similar to those on \lastfmsmall and \lastfm except that the trends are magnified. \welf $\alpha=1$ and \equalexpo seem more similar, with \welf $\alpha=0$ dominating the trade-off total utility/item iniequality when item inequality is measured with the Gini index, and \equalexpo dominating this trade-off when item inequality is measured with standard deviation. \welf $\alpha=-2$ has great performance on worse-off users compared to \equalexpo or \welf with larger $\alpha$, but also comes at a significant cost in terms of total user utility, which is very rapidly driven down.

\begin{figure}
    \centering
    \begin{subfigure}[b]{0.23\textwidth}
         \centering
         \includegraphics[width=\linewidth]{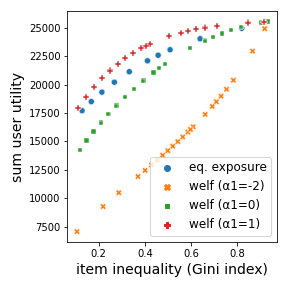}
         \caption{Total user utility.\\~}
     \end{subfigure}\begin{subfigure}[b]{0.23\textwidth}
         \centering
         \includegraphics[width=\linewidth]{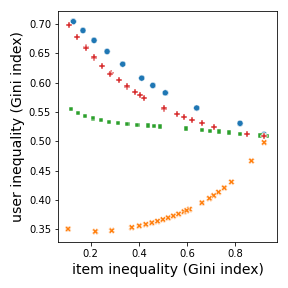}
         \caption{User inequality.\\~}
     \end{subfigure}\begin{subfigure}[b]{0.23\textwidth}
         \centering
         \includegraphics[width=\linewidth]{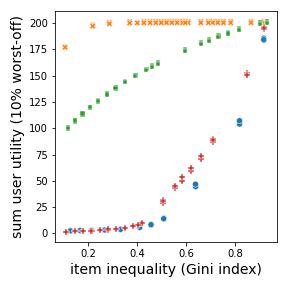}
         \caption{Cumulative utility for 10\% worse-off users}
     \end{subfigure}
     \begin{subfigure}[b]{0.23\textwidth}
         \centering
         \includegraphics[width=\linewidth]{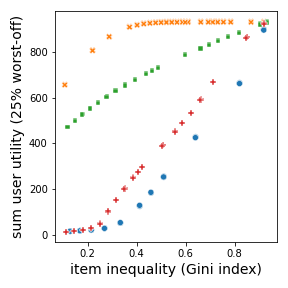}
         \caption{Cumulative utility for 25\% worse-off users}
     \end{subfigure}
    \caption{Results on \mlm when measuring the inequality between items with the Gini coefficient.}
    \label{fig:appendix-lastfm360k-userfairness}
\end{figure}

\begin{figure}
    \centering
    \begin{subfigure}[b]{0.23\textwidth}
         \centering
         \includegraphics[width=\linewidth]{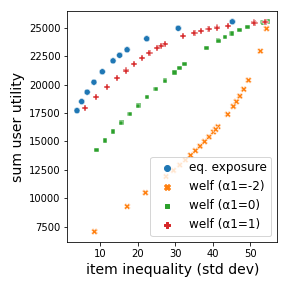}
         \caption{Total user utility.\\~}
     \end{subfigure}\begin{subfigure}[b]{0.23\textwidth}
         \centering
         \includegraphics[width=\linewidth]{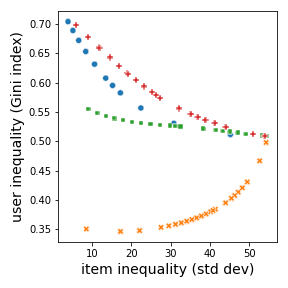}
         \caption{User inequality.\\~}
     \end{subfigure}\begin{subfigure}[b]{0.23\textwidth}
         \centering
         \includegraphics[width=\linewidth]{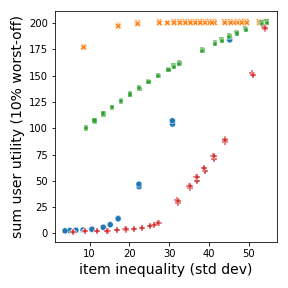}
         \caption{Cumulative utility for 10\% worse-off users}
     \end{subfigure}
     \begin{subfigure}[b]{0.23\textwidth}
         \centering
         \includegraphics[width=\linewidth]{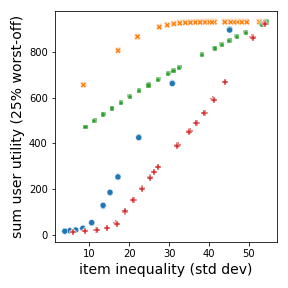}
         \caption{Cumulative utility for 25\% worse-off users}
     \end{subfigure}
    \caption{Results on \mlm when measuring inequalities between  items with the standard deviation.}
    \label{fig:appendix-ml20m-userfairness-std}
\end{figure}

\subsection{Reciprocal recommendation: \twitter dataset}
\label{app:experiments:reciprocal}

We now provide the full details of the experiments on Twitter presented in Section \ref{sec:experiments:reciprocal} of the main body. Given the lack of common benchmark for reciprocal recommendation \citep{palomares2021reciprocal}, we generate a reciprocal recommendation task for people-to-people recommendation problems based on the social network Twitter. We use the Higgs \twitter dataset which includes (directed) follower relationships between users.\footnote{It was collected following the discovery of the Higgs boson in July, 2012.} We keep users having at least $20$ mutual follows, resulting in a subset of $13$k users. We use the directed links to estimate the probability $\phi_{ij}$ that $i$ follows $j$, and the (symmetric) probability of a mutual follow, which is $\mu_{ij}=\phi_{ij} \times \phi_{ji}$.  As in the experiments for one-sided recommendation, we split the dataset into train/validation/test sets, each including $70\%/10\%/20\%$ of the \emph{directed} follower links. We create three random uniform splits, corresponding to three different seeds.

Estimates $\hat{\phi}_{ij}$ are built with logistic matrix factorization\footnote{Using the Python library Implicit (MIT License).} \citep{johnson2014logistic} trained on the train set with hyperparameter selection on the validation set. The number of latent factors is chosen in $[16, 32, 64, 128]$, the regularization in $[0.1, 1., 10., 20., 50.]$, and the confidence weighting parameter in $[0.1, 1., 10., 100.]$. Rankings are inferred from all estimated mutual preferences $\hat{\mu}_{ij} = \max(\hat{\phi}_{ij} \hat{\phi}_{ji}, 0)$. For each ranking method, the Frank-Wolfe algorithm is run with $5000$ iterations, and the number of recommendation slots is set to $40$. As for one-sided recommendation, rankings are estimated on estimated mutual preferences taken as ground truth. 


We generate different trade-offs with \welf by varying $\alpha$ in $[0.99,0.9,0.75,0.5,0.25,0,-0.25,-0.5,-0.6,-0.7,-0.8,-0.9,-1.0]$, $[-1.1,-1.25,-1.5,-1.75,-2.0,-2.5,-3,-5,-10,-15,-16,-17,-18]$. For all other methods, we vary $\beta$ in $[0.01,0.1,0.2,0.3,0.4,0.5,0.6,0.7,0.8,0.9,1,1.1,1.25,1.5,2,5,10,50,100]$.

All presented results are obtained by averaging performance over the three seeds.

\paragraph{Results} Figure \ref{fig:appendix_twitter_all_trade_offs} presents the trade-offs achieved by the different methods on \twitter. As expected, \qualexpo and \equalexpo do not exhibit a good behavior: stronger penalties lead to more dominated curves where the utility of every user is decreased. This is because constraining item exposure is not meaningful in reciprocal recommendation, where the relevant utility is the two-sided utility. The trade-offs achieved by the \welf and \equalu are different. Equal utility rapidly generates near-flat curves without really focusing on the very first users, while \welf increases the utility of the worst-off users while keeping the total utility relatively high.

\begin{figure}
    \centering
    \begin{subfigure}[b]{0.23\linewidth}
    \centering
    \includegraphics[width=\linewidth]{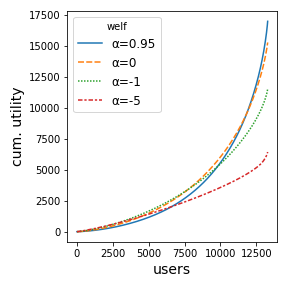}
    \caption{Welfare (ours)}
    \end{subfigure}
    \begin{subfigure}[b]{0.23\linewidth}
    \centering
    \includegraphics[width=\linewidth]{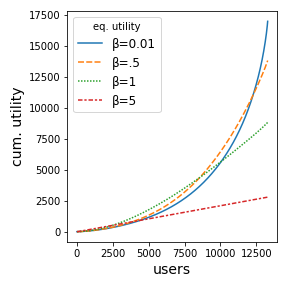}
    \caption{\equalu}
    \end{subfigure}
    \begin{subfigure}[b]{0.23\linewidth}
    \centering
    \includegraphics[width=\linewidth]{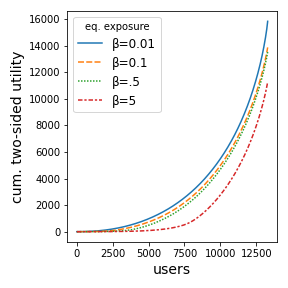}
    \caption{\equalexpo}
    \end{subfigure}
    \begin{subfigure}[b]{0.23\linewidth}
    \centering
    \includegraphics[width=\linewidth]{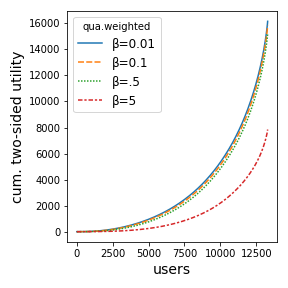}
    \caption{\qualexpo}
    \end{subfigure}
    \caption{representative trade-offs achieved by the various compared methods on \twitter. Exposure-based approaches (\qualexpo and \equalexpo) do not yield interesting trade-offs as they are unable to increase the utility of worse-off users. The trade-offs achieved by the \welf and \equalu are different. Equal utility rapidly generates near-flat curves without really focusing on the very first users, while \welf increases the utility of the worst-off users while keeping the total utility relatively high.}
    \label{fig:appendix_twitter_all_trade_offs}
\end{figure}

\subsection{Reciprocal recommendation: \epinions dataset}\label{sec:epinions}

We present additional experiments on reciprocal recommendation with the Epinions dataset \cite{richardson2003trust}. Epinions.com is a consumer review site with a who-trust-whom network, and the dataset gathers (directed) trust relationships between members of the platform. Here, we consider the task of finding mutual trust links. We keep users having at least $20$ mutual trust links, resulting in a subset of $800$ entities. For the inference and evaluation of rankings, we use the same protocols as for the Twitter experiments described in the previous subsection. The experimental parameters are the same as for the \twitter dataset.

\begin{figure}
    \centering
    \begin{subfigure}[b]{0.23\linewidth}
    \centering
    \includegraphics[width=\linewidth]{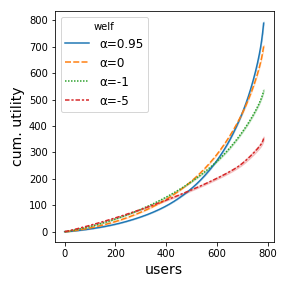}
    \caption{Welfare (ours)}
    \end{subfigure}
    \begin{subfigure}[b]{0.23\linewidth}
    \centering
    \includegraphics[width=\linewidth]{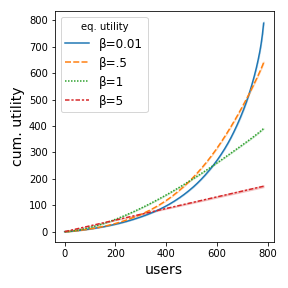}
    \caption{\equalu}
    \end{subfigure}
    \begin{subfigure}[b]{0.23\linewidth}
    \centering
    \includegraphics[width=\linewidth]{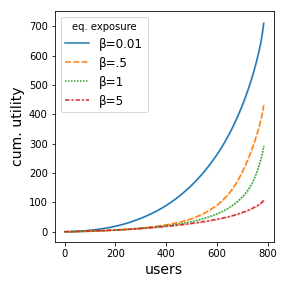}
    \caption{\equalexpo}
    \end{subfigure}
    \begin{subfigure}[b]{0.23\linewidth}
    \centering
    \includegraphics[width=\linewidth]{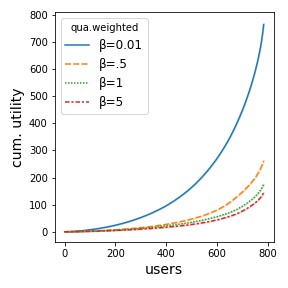}
    \caption{\qualexpo}
    \end{subfigure}
    \caption{representative trade-offs achieved by the various compared methods on \epinions.  The results are qualitatively similar to those on \twitter.}
    \label{fig:appendix_epinions_all_trade_offs}
\end{figure}

\begin{figure}[t]
\centering
    \begin{subfigure}[b]{0.23\textwidth}
         \centering
         \includegraphics[width=\linewidth]{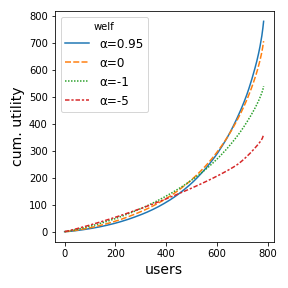}
         \caption{Example trade-offs achieved by \welf.}
         \label{fig:epinions-trafe-offs-welf}
     \end{subfigure}~~
     \begin{subfigure}[b]{0.23\textwidth}
         \centering
         \includegraphics[width=\linewidth]{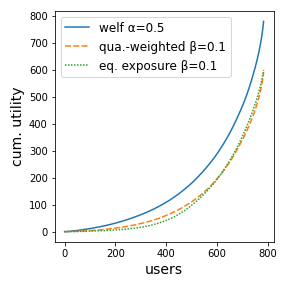}
         \caption{\welf dominates exposure-based methods.}
         \label{fig:epinions-welf-dom-quaexpo}
     \end{subfigure}~~
     \begin{subfigure}[b]{0.23\textwidth}
         \centering
         \includegraphics[width=\linewidth]{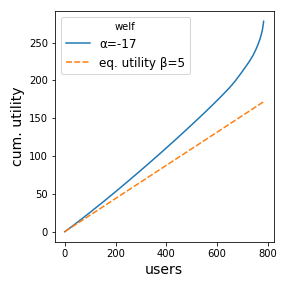}
         \caption{\welf dominates eq. utility near strict equality.}
         \label{fig:epinions-welf-dom-util}
     \end{subfigure}~~

    \begin{subfigure}[b]{0.23\textwidth}
         \centering
         \includegraphics[width=\linewidth]{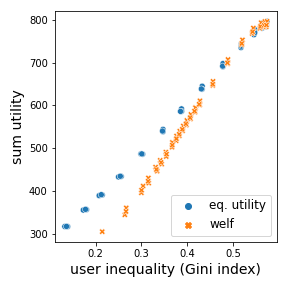}
         \caption{Total utility vs inequality.}
         \label{fig:epinions-uu-vs-ineq}
     \end{subfigure}~~
     \begin{subfigure}[b]{0.23\textwidth}
         \centering
         \includegraphics[width=\linewidth]{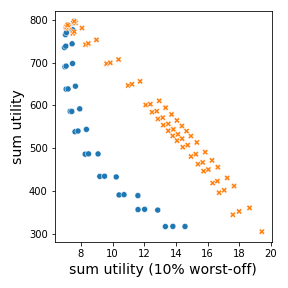}
         \caption{Total utility vs utility of 10\% worse-off users.}
         \label{fig:epinions-uu-first10}
     \end{subfigure}~~
     \begin{subfigure}[b]{0.23\textwidth}
         \centering
         \includegraphics[width=\linewidth]{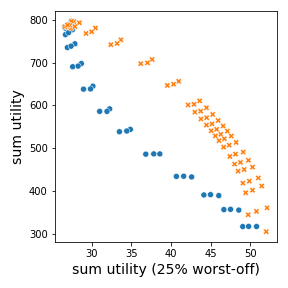}
         \caption{Total utility vs utility of 25\% worse-off users.}
         \label{fig:epinions-uu-first25}
     \end{subfigure}~~
     \begin{subfigure}[b]{0.23\textwidth}
         \centering
         \includegraphics[width=\linewidth]{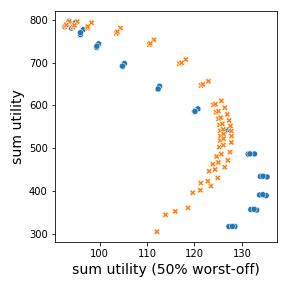}
         \caption{Total utility vs utility of 50\% worse-off users.}
         \label{fig:epinions-uu-first50}
     \end{subfigure}
    \caption{Results on the epinions dataset. }
    \label{fig:epinions-uu}
\end{figure}

\paragraph{Results} Figure \ref{fig:appendix_epinions_all_trade_offs} presents the trade-offs achieved by the different methods on \epinions. As expected, \qualexpo and \equalexpo do not exhibit a good behavior: stronger penalties lead to more dominated curves where the utility of every user is decreased. In Figure~\ref{fig:epinions-uu} plots the equivalent of Fig.~\ref{fig:twitter-uu}. The results are similar: all of \qualexpo, \equalexpo and \equalu have dominated curves. We also observe that in the more interesting region where we are closer to the maximum achievable utility, \welf optimizes better the utility of worse-off users. Yet, in that region, there is no strict dominance of \welf over \equalu.

\section{Pairwise vs pointwise penalties}\label{sec:pairwisepointwise}

Our penalty-based approach uses the penalty $\sqrt{\Penu}$ with:
\begin{equation}
    \Penu= 
    \sumI \Big( \bu_j - 
    \frac{1}{\nitems} \sum_{j'\in\itemS}\bu_{j'} \Big)^2
    .
\end{equation}

Some authors use $\Penpu=\sum_{(j,j')\in\itemS^2} |\bu_j - \bu_j'|$ instead of $\sqrt{\Penu}$ \citep{singh2019policy,morik2020controlling,basu2020framework}, but it is less computationally efficient than our penalty because it involves a quadratic number of terms.

The penalties are similar in that they are related to well-known measures of inequalities:
\begin{itemize}[leftmargin=*]
    \item $\frac{\Penpu}{2\nitems\sum_{j\in\itemS} \bu_j}$ is the Gini index of $\buitems$ \citep{gini1921measurement}, which, up to an affine transform is the area under the Lorenz curve.
    \item $\Penu$, which is (up to a constant) the variance of $\buitems$ is part of the family of additively decomposable inequality measures \citep{shorrocks1980class}. We use $\sqrt{\Penu}$ to scale the penalty with the sum of users' utilities. 
\end{itemize}
Note that $\sqrt{\Penu}$ and $\Penpu$ have the same dependency to the overall scale of the utilities (i.e., multiplying all utilities by a constant factor has the effect of multiplying both penalties by the same factor). Since both penalties drive towards equality, it is straightforward to show that the results of Section~\ref{sec:fairness_exposure} as $\beta\to\infty$ also apply to $\Penpu$. 

\section{Exposure constraints at the level of every ranking}\label{sec:micro}

The notions of fairness of exposure are sometimes defined with item-side constraints defined at the level of \emph{every ranking} \citep{singh2018fairness,basu2020framework}. We give here the examples of constraints for equality of exposure and quality-weighted exposure:
\begin{align}
    \substack{\text{equality of}\\\text{exposure}} 
    &&
    \rrk^\expo \in \argmaxP \sum_{i\in\userS} \tuiP && \text{u.c.~~} \forall (i,j)\in\userS\times\itemS, \rrkij\v = \frac{\norm{\v}_1}{\nitems}\\
    \substack{\text{quality-weighted}\\\text{exposure}} 
    &&
    \rrk^\gqua \in \argmaxP \sum_{i\in\userS} \tuiP && \text{u.c.~~} \forall (i,j)\in\userS\times\itemS, \rrkij\v = \frac{\muij\norm{\v}_1}{\sum_{j'\in\itemS}\mu_{ij'}}\\
\end{align}
The advantage of this formulation is that it leads to optimization problems that can be solved locally for every user, since there is no dependency between user rankings through item utility anymore. 

However, applying the exposure criterion at the level of every ranking effectively applies a different notion of fairness. In our setting, this corresponds to defining a separate recommendation task for every user, i.e., taking $\nusers=1$. The welfare function then mediates, within a single ranking, between the user utility and the utility of the different items.

When evaluated on exposures aggregated over all users, as we do in the paper, applying the fairness constraints at the level of individual rankings can lead to drastic reductions of user utility for no benefit in terms of total item exposure. This is summarized in the following result, which shows that there exists problems for which the optimal rankings for every $\theta\in\Theta$ satisfy the constraints of equality of exposure and quality-weighted exposure as we define them in Section~\ref{sec:fairness_exposure}, but when applying the constraints at the level of every ranking, it has the effect of reducing user utility. In the proposition, we use the notation of the objective function for parity of exposure $F_\beta$ and $\Objgquabeta$ of Section~\ref{sec:fairness_exposure}.
\begin{proposition}
For every $d\in\mathbb{N}_*$ and every $N\in\mathbb{N}_*$, there is a one-sided recommendation task with $d+1$ items and $N(d+1)$ users such that,
$\forall\theta\in\Theta$:

$\forall \tutheta\in\argmaxU\Wthetau$, $\forall\beta>0$
we have: $\tutheta\in\argmaxUl \Objutilbetau$ and $\tutheta\in\argmaxUl  \Objgquabetau$, and
\begin{align}
    \sum\limits_{i\in\userS} \tu_i(\rrk^\expo) = \frac{2}{d+1} \sum\limits_{i\in\userS} \tuthetai && \text{~and~} && \sum\limits_{i\in\userS} \tu_i(\rrk^\gqua) = (\frac{1}{2}+\frac{1}{d}) \sum\limits_{i\in\userS} \tuthetai.
\end{align}
\end{proposition}
In other words, applying the constraints at the level of every ranking might lead to a drastic decrease of user utilities, even in tasks where satisfying the constraints on average over users (as we do in this paper) does not conflict with the optimal ranking.
\begin{proof}
We describe the problem with $N=1$, the general case is obtained by repeating the preference pattern. Let us consider a task with $d+1$ users, $d+1$ items and a single recommendation slot. Let $i_1, \ldots, i_{d+1}$ be the user indexes, and $j_1, \ldots, j_{d+1}$ the item indexes. The preferences are defined as:
\begin{align}
    \forall k\in\intint{d+1}, \mu_{i_kj_k} = 1 && \forall j\neq j_k, \mu_{i_k j} = \frac{1}{d}.
\end{align}
All items have the same quality. For every $\theta\in\Theta$, $\tutheta$ is given by the utilitarian ranking, which gives probability $1$ to item $j_k$ for user $i_k$, which leads to optimal user utility $\tuthetai = 1$ and equal exposure to every item $\tu_j^\theta=1$. Since the quality is the same for all items (equal to $1+d\frac{1}{d}$), the ranking for $\tutheta$ satisfies both equality of exposure and quality-weighted exposure constraints. Thus, for every $\beta>0$,  $\tutheta\in\argmaxU \Objutilbetau$ and $\tutheta\in\argmaxU  \Objgquabetau$.

On the other hand, satisfying equality of exposure at the level of every ranking requires $\rrk^\expo_{ij} = \frac{1}{d+1}$ for every user $i$ and item $j$, which leads to $\tui(\rrk^\expo) = \frac{1}{d+1}+d\times\frac{1}{d}\times\frac{1}{d+1} = \frac{2}{d+1}$ for every user.

For quality-weighted exposure for every ranking, it leads to:
\begin{align}
    \forall k\in\intint{d+1}, \rrk^\gqua_{i_kj_k} = \frac{1}{2} && \forall j\neq j_k,  \rrk^\gqua_{i_ k j} = \frac{1}{d}
\end{align}
and thus a user utility $\tu_i(\rrk^\gqua) = \frac{1}{2} + d\times\frac{1}{d}\times \frac{1}{d} = \frac{1}{2} +\frac{1}{d}$.
\end{proof}
Notice that in the examples of the proof, the global exposure of items is constant in $\rrk^\expo$ and $\rrk^\gqua$, as well as in the ranking given by optimal welfare. So from the point of view of our definitions of utility, applying the constraints at the level of every ranking only decreased user utility for the benefit of no items. Yet, we re-iterate that applying item-side fairness at the level of every ranking might be meaningful in some contexts. The goal of this section is to highlight the difference between using global and local definitions of item utilities.

\end{document}